\documentclass[a4paper,10pt]{article}
\pdfoutput=1 

\usepackage{jheppub} 
\usepackage{amssymb,amstext,amsthm,amsfonts,mathrsfs,amsbsy,amsmath,array,multirow}
\usepackage[all]{xy} \CompileMatrices
\usepackage[T1]{fontenc} 
\usepackage{slashed} 

\newtheorem{thm}{Theorem}[section]

\newtheorem{prop}[thm]{Proposition}
\newtheorem{lemma}[thm]{Lemma}

\newtheorem{definition}[thm]{Definition}

\newtheorem{remark}[thm]{Remark}
\newtheorem{ep}[thm]{Example}

\newcommand{\cCi}{C^\infty}

\newcommand{\la}{\langle}
\newcommand{\ra}{\rangle}

\newcommand{\RR}{{\mathbb R}}

\renewcommand{\(}{\left(}
\renewcommand{\)}{\right)}
\newcommand{\vol}{\mathrm{vol}}

\newcommand{\surj}{\to\kern-1.8ex\to}

\title{\boldmath Self-dual generalized metrics for pure $\mathcal{N}=1$ six-dimensional Supergravity}

\author{M. Garcia-Fernandez and C. S. Shahbazi}

\affiliation{Instituto de Ciencias Matem\'aticas, Madrid.}

\affiliation{Institut de Physique Th\'eorique, CEA-Saclay.}

\emailAdd{mario.garcia@icmat.es}

\emailAdd{carlos.shabazi-alonso@cea.fr}

\abstract{\vspace{0.1cm}\\ 

We geometrize six-dimensional pure $\mathcal{N}=1$ Supergravity by means of an exact Courant algebroid, whose \u{S}evera class is defined through the Supergravity three-form $H$, equipped with a generalized metric and a compatible, torsion-free, generalized connection. The Supergravity equations of motion follow from the vanishing of the Ricci curvature of the generalized metric, satisfying a natural notion of self-duality. This way, we interpret the solutions of six-dimensional pure, $\mathcal{N}=1$, Supergravity as generalized self-dual gravitational monopoles. For the D1-D5 black string solution, we explore the possibility of controlling space-time singularities by using $B$-field transformations.
}

\arxivnumber{1505.03088}

\begin{document}
\maketitle
\flushbottom


\section{Introduction}
\label{sec:introduction}


It is a remarkable fact that some Supergravity theories can be geometrized using different variations of Hitchin's theory of generalized geometry \cite{2002math......9099H}, in the sense that at least their bosonic sector  can be written in terms of natural geometric structures, namely, \emph{generalized metrics}, compatible \emph{generalized connections} and \emph{generalized curvature} quantitities attached to them. The geometrization of a Supergravity theory provide us with a very convenient and natural set-up where to study its structure, its supersymmetric solutions \cite{Coimbra:2014uxa,Coimbra:2015nha}, and in particular its moduli space \cite{delaOssa:2014cia,Anderson:2014xha,2015arXiv150307562G}, being also a unifying and elegant framework interesting from the mathematical point of view \cite{2004math......1221G,2004math......1221G}. So far, the \emph{geometrization program} has been carried out, at least partially, for the higher dimensional Supergravities: for example reference \cite{Coimbra:2011nw} deal with Type-II Supergravities, references \cite{2014CMaPh.332...89G,Coimbra:2014qaa,2015arXiv150307562G} treat Heterotic Supergravity and some of its $\alpha^{\prime}$-corrections and references \cite{Pacheco:2008ps,Coimbra:2011ky} deal with eleven-dimensional Supergravity, where the geometrization has only been achieved for a particular class of backgrounds. In the context of double field theory, the geometrization of supergravity theories was pioneered in \cite{Siegel:1993bj,Siegel:1993th} and later completed in \cite{Jeon:2010rw,Hohm:2010pp,Hohm:2010xe}. For Type-II theories, a geometrization was achieved in \cite{Hohm:2011zr,Hohm:2011dv}, while Heterotic theories and their $\alpha'$-corrections have been described in \cite{Hohm:2011ex,Hohm:2013jaa,Hohm:2014eba}.

A common feature of the geometrization of higher dimensional supergravities is that the bosonic fields can be encoded in a generalized metric $G$ (in absence of RR fields in Type-II), with compatible torsion-free generalized connection, and the equations of motion are described in terms of the \emph{generalized vacuum Einstein equations}
$$
\mathrm{GRic} = 0, \qquad \mathrm{GS} = 0,
$$
where GRic and GS are the natural analogues, in the realm of generalized geometry, of the Ricci tensor and the scalar curvature in Einstein gravity. The main goal of this work is to show that other interesting lower-dimensional Supergravities can be 
studied using generalized geometry. In this letter 
we describe six-dimensional, pure, $\mathcal{N}=1$ Supergravity 
by means of an exact Courant algebroid, whose \u{S}evera class is given by the three-form $H$ of the Supergravity theory. We will see 
how the self-dual condition on $H$ can be naturally reinterpreted in the language of generalized geometry as a self-duality condition for a Lorentz generalized metric in six dimensions, reminiscent of the standard self-duality condition for 
gravitational instantons in four dimensions. With the right notion of self-duality at hand, the Supergravity equations of motion follow from the generalized vacuum Einstein equations. 
On the other hand, imposing the generalized vacuum Einstein equations without assuming self-duality for the generalized metric, provides a class of bosonic theories with the same bosonic content as six-dimensional pure, $\mathcal{N}=1$ Supergravity, such that the three-form $H$ is not self-dual but a section of an appropriate Lagrangian subbundle of the $SO(10,10)$-bundle of three-forms. The interpretation of this class of theories is unclear, but it would be interesting to check whether they can be supersymmetrized in some sense, since they have 
the right number of bosonic degrees of freedom to be supersymmetric. 

To explore the role played by singular solutions in generalized geometry, 
 we consider the generalized metric corresponding to a specific supersymmetric, singular, solution of six-dimensional pure $\mathcal{N}=1$ supergravity: the self-dual string. By constructing an exact Courant algebroid over the singularity of the black string, we will argue that there exists a smooth, degenerate, transition from the black string solution to the solution corresponding to a naked singularity. This is achieved via the construction of a family of generalized Lorentz metrics, degenerating smoothly along the singularity. We speculate that this new point of view may lead to a better understanding and control of space-time singularities in Supergravity.

Six dimensional, pure, $\mathcal{N}=1$ supergravity is a very interesting theory to study in generalized geometry. Firstly, it is arguably one of the simplest supergravity theories that can be geometrized, and thus it is the perfect arena where to fully understand the geometrization procedure and explore its different possibilities. Through compactification or dimensional reduction of six-dimensional, pure, $\mathcal{N}=1$ Supergravity, one can obtain interesting four-dimensional and three-dimensional Supergravities, so despite its simplicity it is related to interesting theories in lower dimensions. In addition, the solutions of six-dimensional pure, $\mathcal{N}=1$ Supergravity \cite{Gutowski:2003rg,Bena:2011dd} play a very important role in relation to the microstate-geometries/fuzzball proposal \cite{Mathur:2005zp} for the description of the black hole entropies in String Theory, see \cite{Lunin:2001fv,Bena:2005va,Bena:2007kg,Lunin:2002iz,Bena:2015bea,deLange:2015gca} for more details and further references. In particular, understanding their moduli space is of paramount importance in this program and we indeed think that generalized geometry is the right framework where to study this problem.

From a mathematical perspective, this work provides a new interesting geometric structure --- self-dual, Lorentzian, generalized metrics in six dimensions ---, reminiscent of the self-duality condition for 
gravitational instantons in four dimensions.
Given an exact Courant algebroid $E$ on a six dimensional oriented manifold $M$, a Lorentzian generalized metric is given by a reduction of the $SO(6,6)$-bundle of frames of $E$ to
$$
SO(1,5)\times SO(1,5)
$$
Under natural transversality conditions with the cotangent bundle $0 \to T^*M \to E$, a generalized metric has an associated three-form curvature $H$ and provides an identification of the $SO(6,6)$-spinor bundle
$$
S_{6,6} \cong \Lambda^* T^*M,
$$
which is endowed with a Hodge star operator
$$
\star \colon S_{6,6} \to S_{6,6}.
$$
Regarding the three-form curvature as an $SO(6,6)$-spinor $H \in \Lambda^* T^*M$, a \emph{self-dual generalized Lorentzian metric} is defined by the condition
$$
\star H = H.
$$
Self-dual generalized metrics can be understood as weak analogues of self-dual gravitational instantons. Being a natural notion in generalized geometry in six-dimensions, self-duality 
produces an interesting coupling for a classical Lorentz metric $g$ with a closed three-form $H$, which seems to be intimately related with paracomplex geometry. 
On the one hand, N. Hitchin has previously studied the geometry of self-dual three-forms in Lorentz signature in his seminal paper \cite{Hit0} showing that, generically, such a three-form defines a paracomplex structure on open sets in the manifold $M$. On the other hand, when the  aforementioned transversality of the generalized metric is lost, it relates to a paracomplex structure on the manifold, as we show in Example \ref{example}. It would be interesting to explore further this relation. 

The outline of this letter goes as follows. In Section \ref{sec:courantalgebroids} we give a gentle introduction to exact Courant algebroids and generalized metrics. In Section \ref{sec:connections} we consider torsion-free, metric compatible, generalized connections and the associated curvature quantities. Section \ref{sec:GSupergravity} contains the main results of this letter: we apply the formalism of Section \ref{sec:courantalgebroids} and Section \ref{sec:connections} to six-dimensional pure, $\mathcal{N}=1$ Supergravity obtaining its description in the language of generalized geometry and identifying its solutions as generalized self-dual gravitational monopoles. In section \ref{sec:susysolutions} we explicitly obtain the generalized metric corresponding to a simple, singular, supersymmetric solution of six-dimensional pure, $\mathcal{N}=1$ Supergravity, in order to illustrate how the singularities appear in the generalized geometry approach. Finally, appendix \ref{app:linearalgebra} contains some mathematical background.


In this letter, by gravitational monopole we mean simply a Ricci-flat (possibly singular) Lorentzian manifold. In the literature the assumption of having finite energy is often included as part of the definition, but we will not assume this extra condition.


\section{Courant algebroids and admissible generalized metrics}
\label{sec:courantalgebroids}


In this section we introduce, closely following references \cite{2007arXiv0710.2719G,2014CMaPh.332...89G,2015arXiv150307562G}, the mathematical background for the geometrization of six-dimensional $\mathcal{N}=1$ pure Supergravity, namely exact Courant algebroids endowed with the appropriate generalized connections. 


\subsection{Courant algebroids}


Intuitively speaking, a Courant algebroid is an enhancement of the tangent bundle of a smooth manifold $M$ by means of a  vector bundle $E\to M$ equipped with a non-degenerate symmetric bilinear form and a bracket, satisfying a relaxed version of the Lie bracket axioms. 
Here $M$ is a $d$-dimensional manifold. A particular example of Courant algebroid, see example \ref{ep:standardCourant}, was first introduced by T. Courant in reference \cite{DiracManifolds} in order to obtain a unified description of pre-symplectic and Poisson structures in Dirac's theory of constrained mechanical systems. Courant algebroids were then abstractly defined for the first time by Liu, Weinstein and Xu in reference \cite{1995dg.ga.....8013L}, and by know there are several equivalent definitions of Courant algebroids available in the literature. Here we will use the definition given in reference \cite{LettersSevera} by \u{S}evera:

\begin{definition}[\cite{LettersSevera}]
\label{def:Courant}
A Courant algebroid $(E,\la\cdot,\cdot\ra,[\cdot,\cdot],\pi)$ over a manifold $M$ consists of a vector bundle $E\to M$ together with a nondegenerate symmetric bilinear form $\la\cdot,\cdot\ra$ on $E$, a (Dorfman) bracket $[\cdot,\cdot]$ on the sections $\Gamma(E)$, and a bundle map $\pi:E\to TM$ such that the following properties are satisfied, for $e_{1},e_{2},e_{3}\in \Gamma(E)$ and $\phi\in \cCi(M)$:
  \begin{itemize}
  \item[(C1):] $[e_{1},[e_{2},e_{3}]] = [[e_{1},e_{2}],e_{3}] + [e_{2},[e_{1},e_{3}]]$,
  \item[(C2):] $\pi([e_{1},e_{2}])=[\pi(e_{1}),\pi(e_{2})]$,
  \item[(C3):] $[e_{1},\phi e_{2}] = \pi(e_{1})(\phi) e_{2} + \phi[e_{1},e_{2}]$,
  \item[(C4):] $\pi(e_{1})\la e_{2}, e_{3} \ra = \la [e_{1},e_{2}], e_{3} \ra + \la e_{2}, [e_{1},e_{3}]
    \ra$,
  \item[(C5):] $[e_{1},e_{2}]+[e_{2},e_{1}]=\pi^* d\la e_{1},e_{2}\ra$.
  \end{itemize}
\end{definition} 

\noindent
The map $\pi\colon E\to TM$ is usually called the \emph{anchor map}. Notice that given an Courant algebroid $E$, we can always identify $E^{\ast}\simeq E$ by using the bilinear $\la\cdot,\cdot\ra$ and hence we obtain a map 

\begin{equation}
\pi^{\ast}\colon T^{\ast}M\to E\, ,
\end{equation}
dual to $\pi\colon E\to TM$. This is the map appearing in item $C5$ of definition \ref{def:Courant}. The bracket in Definition \ref{def:Courant} goes under the name of Dorfman bracket $[\cdot,\cdot]$. It satisfies the Jacobi identity, namely item C1, but fails to be antisymmetric, and relates to the skew-symmetrized \emph{Courant bracket} $[\![\cdot,\cdot ]\!]$, by
\begin{equation}
[\cdot,\cdot] = [\![\cdot,\cdot ]\!] + \pi^* d\la\cdot,\cdot\ra\, .
\end{equation}
The symbol $d$ denotes 
the de Rahm differential. The definition in the original reference \cite{1995dg.ga.....8013L}, differs from definition \ref{def:Courant} in the bracket used (see also \cite{1999math.....10078R}). 
An explicit example of Courant algebroid is now in order.

\begin{ep}
\label{ep:standardCourant}
The simplest example of Courant algebroid is the {\bf standard Courant algebroid} $E=TM\oplus T^{\ast}M$ over a manifold $M$, equipped with the {\bf standard Dorfman bracket}:

\begin{equation}
\label{eq:scourantbracket}
[v_{1}+\alpha_{1}, v_{2}+\alpha_{2}] = [v_{1}, v_{2}]_{L} + \mathcal{L}_{v_{1}}\alpha_{2} - \iota_{v_{2}} d\alpha_{1}\, , \qquad v_{1}, v_{2}\in\mathfrak{X}(M)\, , \qquad \alpha_{1}, \alpha_{2}\in \Omega^{1}(M)\, ,
\end{equation}

\noindent
and the {\bf standard symmetric pairing}:

\begin{equation}\label{eq:pairing}
\la v_{1} + \alpha_{1}, v_{2}+\alpha_{2} \ra = \frac{1}{2}(\iota_{v_{1}}\alpha_{2} + \iota_{v_{2}}\alpha_{1})\, ,
\end{equation}

\noindent
where $[\cdot, \cdot]_{L}$ denotes the standard Lie bracket on $\mathfrak{X}(M)$. The anchor map $\pi\colon E\to TM$ is simply the obvious projection on the tangent bundle. 
\end{ep}

\noindent
It was noticed in reference \cite{LettersSevera} that one can twist the standard Dorfman bracket by using a closed three-form $H$ as follows:

\begin{equation}
\label{eq:Htwistedbracket}
[v_{1}+\alpha_{1}, v_{2}+\alpha_{2}]_{H} = [v_{1}, v_{2}]_{L} + \mathcal{L}_{v_{1}}\alpha_{2} - \iota_{v_{2}} d\alpha_{1} + \iota_{v_{1}}\iota_{v_{2}}H\, ,
\end{equation}

\noindent
and still obtain a Courant algebroid in $TM\oplus T^{\ast}M$, with the same 
anchor and symmetric product. This way, it is obtained the so-called $H$-twisted standard Courant algebroid. The standard Courant algebroid is, as we will see in a moment, the prototype of an exact Courant algebroid.

\begin{definition}\cite{LettersSevera}
\label{def:exactCourant}
A Courant algebroid $(E,\la\cdot,\cdot\ra,[\cdot,\cdot],\pi)$ over $M$ is exact if and only if the following sequence of vector bundles

\begin{equation}
\label{eq:sequenceexactE}
\xymatrix{
0 \ar[r] & T^{\ast}M \ar[r]^{\pi^{\ast}} & E \ar[r]^\pi & TM \ar[r] & 0
}\, ,
\end{equation}

\noindent
is exact.
\end{definition}

\begin{definition}\cite{LettersSevera}
\label{def:splittingCourant}
A splitting of an exact Courant algebroid $(E,\la\cdot,\cdot\ra,[\cdot,\cdot],\pi)$ over a manifold $M$ is a map of vector bundles $s\colon TM \to E$ such that

\begin{enumerate}

\item $\pi\circ s = \mathbb{I}_{TM}$,

\item $\la s(v_{1}), s(v_{2})\ra = 0$ for all $v_{1}, v_{2} \in \mathfrak{X}(M)$.

\end{enumerate}

\end{definition}

\noindent
Definition \ref{def:splittingCourant} means that a splitting of an exact Courant algebroid is an isotropic splitting of the sequence of vector bundles \ref{eq:sequenceexactE}. Notice that $\pi^{\ast}\left(T^{\ast}M\right)\cap s\left(TM\right) = \left\{ 0\right\}$. The exactness condition in the definition \ref{def:exactCourant} forces $\pi^{\ast}(T^{\ast}M)$ to be isotropic in $E$ and thus the symmetric pairing $\la\cdot,\cdot\ra$ is bound to be of split signature. If $s$ is an splitting, then for every two-form $b\in\Omega^{2}(M)$ we can construct another splitting $s'$ as follows

\begin{equation}
s'(v) = s(v) + \frac{1}{2}\pi^{\ast}b(v)\, ,
\end{equation}

\noindent
and in fact every two splittings of a Courant algebroid differ by a two-form on $M$ in this way \cite{Bressler:2002ur}. In other words, the space of splittings of a Courant algebroid is an affine space modeled on $\Omega^{2}(M)$. Given an exact Courant algebroid $(E,\la\cdot,\cdot\ra,[\cdot,\cdot],\pi)$, any isotropic splitting $s\colon TM\to E$, has an associated \emph{three-form curvature}:

\begin{equation}
H(v_{1},v_{2},v_{3}) = \la [\![ s(v_{1}), s(v_{2})]\!], s(v_{3})\ra\, , \qquad v_{1}, v_{2}, v_{3} \in \mathfrak{X}(M)\, .
\end{equation}

\noindent
It can be proven that given another splitting $s'$ then the corresponding three-form: 

\begin{equation}
H'(v_{1},v_{2},v_{3}) = \la [\![ s'(v_{1}), s'(v_{2})]\!], s'(v_{3})\ra\, , \qquad v_{1}, v_{2}, v_{3} \in \mathfrak{X}(M)\, .
\end{equation}

\noindent
is related to $H$ as follows:

\begin{equation}
H' = H + db\, ,
\end{equation}

\noindent
where $s' - s = b\in\Omega^{2}(M)$. As observed first by \u{S}evera \cite{LettersSevera}, given an exact Courant algebroid $(E,\la\cdot,\cdot\ra,[\cdot,\cdot],\pi)$, the class $[H]\in H^{3}(M)$ does not depend on the splitting. It is called the \u{S}evera class of the exact Courant algebroid and its importance steams from the fact that it classifies exact Courant algebroids up to isomorphism. In other words, two exact Courant algebroids are isomorphic if and only if they have the same \u{S}evera class. 

Notice that for an exact Courant algebroid $(E,\la\cdot,\cdot\ra,[\cdot,\cdot],\pi)$, any isotropic splitting $s \colon TM \to E$ determines an isomorphism
$$
s + \frac{1}{2}\pi^* \colon TM \oplus T^*M \to E,
$$
and the trasported bracket and pairing are given by \ref{eq:Htwistedbracket} and \ref{eq:pairing}, respectively. Therefore, exact Courant algebroids over a manifold $M$ can be always modeled by the corresponding generalized tangent bundle $TM\oplus T^{\ast}M$ equipped with the standard symmetric pairing and the $H$-twisted Courant bracket \ref{eq:Htwistedbracket}.


\subsection{Admissible generalized metrics}


As mentioned earlier, the symmetric pairing $\la \cdot, \cdot\ra$ of an exact Courant algebroid is of signature $(d,d)$. In particular it implies a reduction of the bundle of frames of $E$ to $O(d,d)$. A generalized metric of signature $(p,q)$ is a further reduction of the bundle of frames of $E$ from $O(d,d)$ to:

\begin{equation}
O(p,q)\times O(q,p)\subset O(d,d)\, .
\end{equation}

\noindent
More geometrically, a generalized metric can alternatively be defined as a subbundle

\begin{equation}
V_{+}\subset E\, ,
\end{equation}

\noindent
such that the restriction of the symmetric pairing $\la \cdot, \cdot\ra$ to $V_{+}$ is a non-degenerate metric of signature $(p,q)$. Notice that for the case of an exact Courant algebroid with $p=d$, $q=0$, we recover the definition of generalized metric given in \cite{2004math......1221G}, which induces a Riemannian metric in the corresponding vector bundle $V_{+}$. We will denote by $V_{-}$ the orthogonal complement of $V_{+}$ with respect to the symmetric pairing. A generalized metric is equivalent to the existence of a vector bundle isomorphism

\begin{equation}
G\colon E\to E\, ,
\end{equation}

\noindent
whose eigenspaces are $V_{\pm}$ and such that $G^{2}=\mathbb{I}$, $G^{t} = G$. The subbundle $V_{+}$ can be obtained from $G$ as follows:

\begin{equation}
V_{+} = \mathrm{Ker}(G-\mathbb{I})\, .
\end{equation}

\noindent
Now, there is an important difference between the indefinite-signature case $(p,q)$ and the Riemannian case studied in \cite{2004math......1221G}. In the Riemannian case we automatically have that $V_{+}\cap \pi^{\ast}(T^{\ast}M) = \left\{ 0\right\}$ for any generalized metric $V_{+}$. This fact is key in order to relate generalized metrics to a standard metric in $M$ and a splitting of $E$. For Supergravity applications it is important to keep this relation in the indefinite-signature case, and therefore we define:

\begin{definition}
A generalized metric $V_{+}$ on an exact Courant algebroid $(E,\la\cdot,\cdot\ra,[\cdot,\cdot],\pi)$ is said to be {\bf admissible} if

\begin{equation}
\label{eq:admissiblecondition}
V_{+}\cap\pi^{\ast}( T^{\ast}M) = \left\{ 0\right\}\, .
\end{equation}

\end{definition}

\noindent
When equation \eqref{eq:admissiblecondition} holds, the restriction of the anchor map $\pi_+ \colon V_+ \to TM $ is an isomorphism, which precisely induces a metric of signature $(p,q)$ on $M$ as follows:

\begin{equation}
g(v_{1},v_{2}) = \la \pi_+^{-1}(v_{1}), \pi_+^{-1}(v_{2})\ra\, .
\end{equation}

\noindent
For admissible metrics we can prove the following.

\begin{prop}\label{prop:ghfromVandback}
An admissible metric $V_{+}$ on an exact Courant algebroid $(E,\la\cdot,\cdot\ra,[\cdot,\cdot],\pi)$ is equivalent to a pair $(g, s)$, where $g$ is a metric on $TM$ and $s\colon TM \to E$ is an isotropic splitting such that

\begin{equation}\label{eq:V+}
V_{+} = \left\{ s(v) + \frac{1}{2}\pi^{\ast}g(v)\,\, |\,\, v\in TM\right\}
\end{equation}
\end{prop}

\begin{proof}
Define an splitting $s \colon TM \to E$ by $s(v) = \pi_+^{-1}(v) - \frac{1}{2}\pi^{\ast}g(v)$. Then, by definition of $\pi^*$ and $g$, we have
$$
\langle s(v),s(w)\rangle = g(v,w) - \frac{1}{2}(g(v,w) + g(w,v)) = 0,
$$
and hence $s$ is isotropic. Now, for an arbitrary element $e_+ \in V_+$ we have 
$$
e_+ = \pi_+^{-1}v = s(v) + \frac{1}{2}\pi^{\ast}g(v),
$$
where $v = \pi e_+$, as claimed. The converse follows by a direct check using formula \ref{eq:V+}.
\end{proof}

We note that using the isomorphism $s + \frac{1}{2}\pi^* \colon TM \oplus T^*M \to E$ given by the splitting provided by a generalized metric, the transported subbundle $V_+ \subset TM \oplus T^*M$ has a very simple description
$$
V_{+} = \left\{v + g(v)\,\, |\,\, v\in TM\right\}.
$$
Fixing the isomorphism and moving the splitting by a $b$-field transformation, for $b \in \Omega^2(M)$, we obtain the familiar description of generalized metrics in terms of pairs $(g,b)$
\begin{equation}
e^b V_{+} = \left\{v + b(v) + g(v)\,\, |\,\, v\in TM\right\} \subset E.
\end{equation}
For this generalized metric, the endomorphism $G \colon TM \to TM$ can be explicitely written as
\begin{equation}\label{eq:Gtransverse}
G = e^b \left(
\begin{array}{cc}
0 & g^{-1}\\
g & 0
\end{array}\right) e^{-b} =  \left(
\begin{array}{cc}
-g^{-1}b & g^{-1} \\
g - bg^{-1}b & bg^{-1}
\end{array}\right).
\end{equation}

\noindent
An admissible metric $V_{+}$ on an exact Courant algebroid $(E,\la\cdot,\cdot\ra,[\cdot,\cdot],\pi)$ determines, via the associated splitting, a closed three-form $H$ on $M$ such that the bracket in the splitting $E\simeq TM\oplus T^{\ast}M$ provided by $V_{+}$ is given by \ref{eq:Htwistedbracket}.

\begin{definition}
Given an admissible generalized metric, we define its three-form curvature as the closed three-form $H$ determined by the associated splitting.
\end{definition}

\noindent
We note that regarding $H$ as a curvature for the generalized metric, the closed property
$$
d H  = 0
$$
is naturally interpreted as a Bianchi identity.



We finish this section with an example of a generalized Lorentz metric in two-dimensions, which fails to be admissible on a prescribed locus. As we will see, non-admissible metrics have a close relation with paracomplex structures on the manifold.

\begin{ep}\label{example}
Consider the flat Lorentz metric on $\RR^2$
$$
g_{1,1} = d x^u \otimes d x^v + d x^v \otimes d x^u,
$$ 
and the skew-orthogonal paracomplex structure on $\RR^2$, that is $A \colon T\RR^2 \to T\RR^2$ satisfying 
$$
A^2 = \mathbb{I}_{T\RR^2}, \qquad g_{1,1}(A\cdot,\cdot) + g_{1,1}(\cdot,A \cdot) = 0,
$$
given by
$$
A = \left(
\begin{array}{cc}
1 & 0\\
0 & -1
\end{array}\right)
$$
in coordinates $(x^u,x^v)$. For a non-negative function $f \geq 0$ on $\RR^2$, the expression
\begin{equation}
G_f = \left(
\begin{array}{cc}
A & f g_{1,1}^{-1}\\
0 & A^*
\end{array}\right),
\end{equation}
defines a smooth generalized metric of signature $(1,1)$ on $T\RR^2 \oplus T^* \RR^2$ which fails to be admissible in the whole of $\RR^2$. Note that, in the locus $f \neq 0$ the metric $G_f$ is, nevertheless, admissible: we can define $b_f = - \frac{1}{f}g_{1,1}A$ and hence
\begin{equation}\label{eq:GtransverseII}
G_f = e^{b_f} \left(
\begin{array}{cc}
0 & f g_{1,1}^{-1}\\
f^{-1}g_{1,1} & 0
\end{array}\right) e^{-b_f}.
\end{equation}
By the previous expression, the smooth non-admissible metric $G^f$ on $\RR^2$ is B-field related to the classical, possibly singular, metric $\frac{1}{f}g_{1,1}$.
%
\end{ep}


\section{Torsion-free generalized connections and its curvature}
\label{sec:connections}


In this section we are going to consider generalized connections on Courant algebroids, as introduced by Gualtieri in \cite{2007arXiv0710.2719G}. Given a generalized metric, we will introduce a canonical notion of Levi-Civita connection following \cite{2014CMaPh.332...89G,2015arXiv150307562G} and the relevant curvature quantitites for the analysis of $\mathcal{N} = 1$ Supergravity.


\subsection{The canonical Levi-Civita connection}


A generalized connection $D$ on $E$ is a first order differential operator:

\begin{equation}
D \colon \Gamma(E) \to \Gamma(E^* \otimes E)\, ,
\end{equation}

\noindent
which satisfies the Leibniz rule $D_{e_{1}}(fe_{2}) = fD_{e_{1}}e_{2} + \pi(e_{1})(f)e_{2}$, for $e_{1}, e_{2} \in \Gamma(E)$ and $f \in C^\infty(M)$. In the case of exact Courant algebroids, a generalized connection can be seen as a consistent way of taking derivatives with respect to tangent and cotangent directions, for sections of the generalized tangent bundle. 

For Supergravity applications, we will endow the Courant algebroid $(E,\la\cdot,\cdot\ra,[\cdot,\cdot],\pi)$ with a generalized metric $G$ and we will only consider connections compatible with the symmetric pairing on $E$, that is, satisfying:

\begin{equation}
\pi(e_{1})(\langle e_{2},e_{3} \rangle) = \langle D_{e_{1}} e_{2},e_{3} \rangle + \langle e_{2},D_{e_{1}} e_{3} \rangle\, ,
\end{equation}

\noindent
as well as compatible with the generalized metric $G$, namely:

\begin{equation}
D(Ge) = G(De),\, \qquad e\in\Gamma(E)\, .
\end{equation}

\noindent
We will refer to this kind of connections simply as generalized metric connections. Given a standard connection $\nabla\colon \Gamma(E)\to\Gamma(T^{\ast}M\otimes E)$ compatible with the symmetric pairing $\la\cdot,\cdot\ra$ and the generalized metric $G$ we can construct a generalized metric connection simply by setting:

\begin{equation}
D'_{e} s = \nabla_{\pi(e)} s\, , \qquad s,e\in\Gamma(E)\, .
\end{equation}

\noindent
Any other generalized connection can be written in terms of $D'$ by means of an element

\begin{equation}
\chi\in\Gamma\left(E^{\ast}\otimes (\mathfrak{o}(V_{+})\oplus\mathfrak{o}(V_{-}))\right)\, ,
\end{equation}

\noindent
as follows:

\begin{equation}
D = D' +\chi\, ,
\end{equation}

\noindent
Therefore, the space of metric compatible connections on $E$ is an affine space modelled on the vector space $\Gamma\left(E^{\ast}\otimes (\mathfrak{o}(V_{+})\oplus\mathfrak{o}(V_{-}))\right)$. 

In reference \cite{2007arXiv0710.2719G} a natural notion of connection on a Courant algebroid, the \emph{Gualtieri--Bismut connection}, was introduced in order to characterize generalized K\"ahler geometry. For Supergravity applications, the relevant generalized connection is in fact an analogue of the Levi-Civita connection --- obtained from the Gualtieri--Bismut connection by \emph{killing the torsion} --- introduced in a more general setup in \cite{2014CMaPh.332...89G,2015arXiv150307562G}. We will review now its construction focusing on the case of exact Courant algebroids, which is the relevant scenario for six-dimensional, pure, $\mathcal{N}=1$ Supergravity.

To any given admissible metric $V_{+}$, we can associate an endomorphism of the vector bundle $E$ such that $C(V_{+}) = V_{-}$ and $C(V_{-}) = V_{+}$, defined by:

\begin{equation}
C = \pi_{|V_{-}}^{-1} \circ \pi \circ \Pi_+ + \pi_{|V_{+}}^{-1} \circ \pi \circ \Pi_-\, ,
\end{equation}

\noindent
where

\begin{equation}
\Pi_{\pm} = \frac{1}{2}(\mathbb{I} \pm G) \colon E \to V_{\pm}\, ,
\end{equation}

\noindent
denote the orthogonal projections. A generalized metric $G$ defines a canonical splitting given by the $G$-orthogonal complement of $V_{+}$. Using the canonical splitting provided by $V_{+}$

\begin{equation}
\label{eq:Esplittr}
E \simeq TM \oplus T^{\ast}M\, ,
\end{equation}

\noindent
we have

\begin{equation}
\label{eq:V+b0}
V_{+} = \{v + g(v) \colon v \in TM \}\, ,
\end{equation}

\noindent
and then we can explicitly write

\begin{equation}
C(v + gv) = v- gv \quad \textrm{and} \quad C(u - gu) = u + gu\, , \qquad u,v \in TM\, .
\end{equation}

\begin{definition}[\cite{2007arXiv0710.2719G}]
The {\bf Gualtieri--Bismut connection} $D^B = D^B(V_{+})$ of $V_{+}$ on $E$ is defined by:

\begin{equation}
\label{eq:bismut}
D^B_{e_{1}}e_{2} = [e_{1-},e_{2+}]_+ + [e_{1+},e_{2-}]_- + [C e_{1-},e_{2-}]_- + [Ce_{1+},e_{2+}]_+\, ,
\end{equation}

\noindent
where $e_{i\pm} = \Pi_\pm e_{i},\, i=1,2$.
\end{definition}

\noindent
The generalized torsion (see \cite{2007arXiv0710.2719G}) of a connection $D$ on $E$ is the totally skew tensor $T_D \in \Lambda^3 E^*$ defined by

\begin{equation}
T_D(a,b,c) = \langle D_{a}b - D_{b}a - [\![ a,b ]\!],c \rangle + \frac{1}{2}\(\langle D_{c} a,b \rangle - \langle D_{c} b,a \rangle\)\, .
\end{equation}

\noindent
A generalized connection with vanishing torsion will be refered as a generalized \emph{torsion-free connection}. By analogy with Hermitian geometry, in \cite{2015arXiv150307562G} the first author jointly with Rubio and Tipler, introduced the following notion of Levi-Civita connection associated to a generalized metric. 

\begin{definition}[\cite{2014CMaPh.332...89G,2015arXiv150307562G}]
The {\bf canonical Levi-Civita connection} of $V_{+}$ is defined by

\begin{equation}
\label{eq:LC}
D^{LC} = D^B - \frac{1}{3} T_{D^B},
\end{equation}

\noindent
where we identify the torsion $T_{D^B}$ with the element $\langle \cdot, \cdot \rangle^{-1} T_{D^B} \in E \otimes \Lambda^2 E^*$.
\end{definition}

\noindent
By construction, $D^{LC}$ is a natural object on $E$, that is, given an automorphism $\varphi \colon E \to E$ and a generalized metric $V_{+}$, we have

\begin{equation}
\varphi_{\ast}(D^{LC}(V_{+})) = D^{LC}(\varphi(V_{+}))\, .
\end{equation}

\noindent
As it was done in \cite{2014CMaPh.332...89G}, we can now modify $D^{LC}$ by elements in $E^{\ast}$, while preserving the torsion-free property: hence, torsion-free, metric connections are not unique. However, $D^{LC}$ is, among all the torsion-free generalized metric connections, a canonical, natural, choice.

Let us fix a generalized metric $V_{+}$ on $E$ and consider the associated splitting \eqref{eq:Esplittr}. In this splitting, the generalized metric takes the form
\begin{equation}
V_{+} = \{X + g(X) \colon X \in TM \}\, , \qquad V_{-} = \{X - g(X) \colon X \in TM \}\, ,
\end{equation}

\noindent
and the induced three-form $H$ is closed

\begin{equation}
dH =0\, .
\end{equation}

\noindent
Since they will naturally appear in a moment, let us define now the following connections on $TM$ with skew torsion, compatible with the metric $g$, given by

\begin{equation}
\label{eq:nabla+3}
\nabla^\pm = \nabla^g \pm \frac{1}{2}g^{-1}H\, , \qquad \nabla^{\pm 1/3} = \nabla^g \pm \frac{1}{6}g^{-1}H\, ,
\end{equation}

\noindent
where $\nabla^g$ denotes the Levi-Civita connection of the metric $g$ on $M$. Setting:

\begin{equation}
\label{eq:abcd}
\begin{split}
a_+ &= v + gv\, ,\\
b_- &= u - gu\, ,\\
c_+ &= w + gw\, ,\\
d_- &= x - gx\, ,
\end{split}
\end{equation}

\noindent
where $v, y, w, x \in TM$, we have

\begin{equation}\label{eq:bismutexp}
\begin{split}
D^B_{a_{+}} c_+ &= 2\Pi_+\(\nabla^{+}_{v} w \)\, ,\\
D^B_{b_{-}} c_+ &= 2\Pi_+\(\nabla^{+}_{u} w\)\, ,\\
D^B_{a_{+}} b_- &= 2\Pi_-\(\nabla^{-}_{v} u\)\, ,\\
D^B_{b_{-}} d_- &= 2\Pi_-\(\nabla^{-}_{u} x\)\, .
\end{split}
\end{equation}

\noindent
In reference \cite{2007arXiv0710.2719G}, Gualtieri gave a  formula for the torsion of the generalized connection $D^B$. We provide an alternative derivation of Gualtieri's formula, using the explicit method in \cite{2014CMaPh.332...89G}. Consider the auxiliar covariant derivative:

\begin{equation}
D' = \nabla^g \oplus \nabla^{g^*}\, ,
\end{equation}

\noindent
on $E$, compatible with $V_{+}$. Define:
 
\begin{equation}
\label{eq:GtransverseIII}
\begin{split}
\chi'_e &= - \langle \cdot,\cdot\rangle^{-1}T_{D'} = \left(
\begin{array}{cc}
0 & 0 \\
i_X H  & 0
\end{array}\right) \in \mathfrak{o}(E)\, ,
\end{split}
\end{equation}

\noindent 
for $e = X + \xi$. Then, we have:

\begin{equation}
\label{eq:bismutchi}
D^B = D' + (\chi')^{+++} + (\chi')^{---} + (\chi')^{+-+} + (\chi')^{-+-}\, .
\end{equation}

\noindent
With the previous formula, a direct calculation using \cite[Lemma 3.5]{2014CMaPh.332...89G} lead us to the following expression for the torsion: the torsion $T_{D^B}$ is the element of $\Lambda^3V_{+}^* \oplus \Lambda^3V_{-}^*$ given by

\begin{equation}
\label{eq:bismuttorsion}
T_{D^B} = \pi_{|V_{+}}^*H + \pi_{|V_{-}}^*H\, . 
\end{equation}

\noindent
More explicitely, taking $a_+,c_+$ as in \eqref{eq:abcd} and $b_+ = Y + gY$, we obtain the formula 

\begin{equation}
T^+_{D^B}(a_+,b_+,c_+)  = H(v,u,w)\, ,
\end{equation}

\noindent
for $T_{D^B} = T_{D^B}^+ + T_{D^B}^-$ the natural decomposition. Similarly, setting $a_- = X - gX$ we have

\begin{equation}
T^{-}_{D^B}(a_-,b_-,d_-) = H(v,u,x)\, .
\end{equation}

\noindent
We write now an explicit formula for the Levi-Civita connection of $V_{+}$. We have:

\begin{equation}
\langle \cdot,\cdot\rangle^{-1} T_{D^B} = 2(\chi')^{+++} + 2(\chi')^{---}\, ,
\end{equation}

\noindent
and therefore:

\begin{equation}
\label{eq:LeviCivitachi}
\begin{split}
D^{LC} & = D^B - \frac{2}{3} (\chi')^{+++} - \frac{2}{3} (\chi')^{---}\\
& = D' + \frac{1}{3} (\chi')^{+++} + \frac{1}{3} (\chi')^{---} + (\chi')^{+-+} + (\chi')^{-+-}\, .
\end{split}
\end{equation}

\noindent
Hence, we conclude that:

\begin{equation}\label{eq:Levi-Civitaexp}
\begin{split}
D^{LC}_{a_+} c_+ &= 2\Pi_+\(\nabla^{1/3}_{v} w\)\, ,\\
D^{LC}_{b_-} c_+ &= 2\Pi_+\(\nabla^+_{u} w\)\, ,\\
D^{LC}_{a_+} b_- &= 2\Pi_-\(\nabla^-_{v} u\)\, ,\\
D^{LC}_{b_-} d_- &= 2\Pi_-\(\nabla^{-1/3}_{u} x\)\, .
\end{split}
\end{equation}

\begin{remark}
\label{rem:D0}
Note that the connection $D^{LC}$ equals the torsion-free connection $D^0$ constructed in \cite{2014CMaPh.332...89G}, specialized to an exact Courant algebroid.
\end{remark}


\subsection{Generalized curvature}
\label{sec:generalizedcurvature}


In this section we provide formulae for the generalized curvature, generalized Ricci tensor and generalized scalar curvature of an admissible metric on the exact Courant algebroid $E$, with respect to the canonical Levi-Civita connection $D^{LC}$. 

The generalized curvature of a generalized connection $D$ was defined by Gualtieri in \cite{2007arXiv0710.2719G} as follows: 

\begin{equation}
GR(e_1,e_2) = D_{e_1}D_{e_2} - D_{e_2}D_{e_1} - D_{[[e_1,e_2]]} \in \mathfrak{o}(E)\, ,
\end{equation}

\noindent
for $e_1,e_2 \in C^\infty(E)$. This quantity becomes tensorial when evaluated on a pair of orthogonal sections. In particular, given a generalized metric

\begin{equation}
E = V_{+} \oplus V_{-}\, ,
\end{equation}

\noindent
we obtain a tensor by restriction

\begin{equation}
GR \in V_{+}^{*} \otimes V_{-}^{*} \otimes \mathfrak{o}(E)\, .
\end{equation}

\noindent
Let us fix an admissible metric $V_{+} \subset E$, with corresponding standard metric $g$ and consider the torsion free, compatible, generalized connection $D^{LC}$. We recall the explicit calculation of its curvature from \cite{2014CMaPh.332...89G}

\begin{equation}
GR(a_+,b_-)c_+ \in V_{+},
\end{equation}

\noindent
for $a_+,c_+ \in V_{+}$, and $b_- \in V_{-}$.  Using the splitting $E \cong T \oplus T^*$ provided by $V_{+}$ we obtain:

\begin{equation}
D^{LC}_{a_+}D^{LC}_{b_-}c_+  = 2\Pi_{+}\left(\nabla_v^{1/3}\nabla^+_u w\right)\, ,
\end{equation}

\begin{equation}
D^{LC}_{b_-}D^{LC}_{a_+}c_+  = 2\Pi_{+}\left(\nabla_u^+\nabla^{1/3}_v w\right)\, .
\end{equation}

\noindent
Using the equality

\begin{equation}
[[a_+,b_-]] = [v,u] - g(\nabla^g_vu + \nabla^g_uv,\cdot) + H(v,u,\cdot)\, ,
\end{equation}

\noindent
we also obtain:

\begin{align*}
D^{LC}_{[[a_+,b_-]]}c_+ & =  2\Pi_+\(\nabla^g_{[v,u]}w\)+ \frac{1}{3}\chi^{+++}_{[[a_+,b_-]]}c_+ + \chi^{+-+}_{[[a_+,b_-]]}c_+\\
& = 2\Pi_+\Bigg{(}\nabla^g_{[v,u]}w + \frac{1}{6}g^{-1}H(2[v,u] + \nabla^g_vu + \nabla^g_u v,w,\cdot)\\
&  - \frac{1}{6}g^{-1}H(g^{-1}H(v,u,\cdot),w,\cdot)\Bigg{)} + \nabla^g_vu + \nabla^g_uv - g^{-1}H(v,u,\cdot)\, .
\end{align*}

\noindent
Finally, the identity:

\begin{equation}
(\nabla_X^{g}H)(u,w,\cdot) =  \nabla_X^g(H(u,w,\cdot)) - H(\nabla^g_vu,w,\cdot) - H(u,\nabla^gw,\cdot)\, ,
\end{equation}

\noindent
lead us to the following expression for the curvature:

\begin{lemma}[\cite{2014CMaPh.332...89G}]
\label{lem:curvature}
\begin{equation}
\begin{split}
GR(a_+,b_-)c_+  & = 2\Pi_+\Bigg{(}R^g(v,u)w + g^{-1}\Bigg{(}\frac{1}{2}(\nabla^g_vH)(u,w,\cdot) - \frac{1}{6}(\nabla^g_uH)(v,w,\cdot)\\
& + \frac{1}{12}H(v,g^{-1}H(u,w,\cdot),\cdot) - \frac{1}{12}H(u,g^{-1}H(v,w,\cdot),\cdot)\\
& - \frac{1}{6}H(w,g^{-1}H(v,u,\cdot),\cdot)\Bigg{)}\,\Bigg{)}\,
\end{split}
\end{equation}
\end{lemma}
\noindent

As noticed first by Gualtieri and Hitchin, there is a natural Ricci tensor associated to any generalized metric connection
\begin{equation}
GR \in V_{-}^{*} \otimes V_{+}^{*}.
\end{equation}
In particular, attached to any admissible metric $G$ there is a natural Ricci tensor \cite[Sec. 4.2]{2014CMaPh.332...89G}, defined via the the canonical Levi-Civita connection $D^{LC}$. Acting on $b_-, c_+$, this is defined as the trace of the endomorphism:

\begin{equation}
a_+ \to GR(a_+,b_-)c_+.
\end{equation}

\noindent
To provide a formula for $GR$, define the tensor:

\begin{equation}
H \circ H (u,w) \equiv \sum_{j=1}^n g(H(e_j,u,\cdot),H(e_j,w,\cdot))\, ,
\end{equation}

\noindent
where $\{e_j\}^n_{j=1}$ is an orthonormal basis for $g$. Notice that, in coordinates,

\begin{equation}
(H \circ H)_{ij} = H_{ikl}H_j^{\vspace{0.1cm}kl}\, .
\end{equation}

\noindent
Using $H \circ H$, we can now write the Ricci tensor of the connection $\nabla^+$ as

\begin{equation}
\mathrm{Ric}^{+} = \mathrm{Ric}^g - \frac{1}{4} H \circ H - \frac{1}{2}d^*H\, ,
\end{equation}

\noindent
where $\mathrm{Ric}^g$ denotes the Ricci tensor of $g$. Recall that the adjoint of the exterior differential of a $k$-form $\beta$ can be calculated as

\begin{equation}
d^*\beta = - \sum_{j=1}^n i_{e_j}(\nabla^{g*}_{e_j}\beta)
\end{equation}

\noindent
Then, as a straightforward consequence of Lemma \ref{lem:curvature}, we obtain the desired expression for the generalized Ricci tensor.

\begin{prop}[\cite{2014CMaPh.332...89G}]\label{prop:Ricci}
\begin{equation}
GR(b_-,c_+)  = \mathrm{Ric}^+(u,w)\, .
\end{equation}
\end{prop}

\noindent
The very definition of the generalized Ricci tensor, as an element

\begin{equation}
GR \in V_{-}^{*} \otimes V_{+}^{*}\, ,
\end{equation}

\noindent
implies that there is no natural scalar quantity associated to it. To circumvent this problem in a different context, it was proposed in \cite{Coimbra:2011nw} to associate a generalized scalar curvature:

\begin{equation}
GS \in C^\infty(M)\, ,
\end{equation}

\noindent
using `squares of Dirac operators' naturally associated to a torsion-free generalized connection. In our setup, this is done as follows: by compatibility with $G$, $D^{LC}$ induces differential operators:

\begin{equation}
D^{LC}_{\pm}\colon V_{+}\to V_{+}\otimes V_\pm^{\ast}\, .
\end{equation}

\noindent
Assuming the spin condition $w_{2}(TM) = 0$ for the manifold $M$, there exist spinor bundles $S(V_{\pm}) = S_{+}(V_{\pm})\oplus S_{-}(V_{\pm})$. From $D^{LC}_{+}$ and $D^{LC}_{-}$ we get differential operators on spinors:

\begin{equation}
\label{eq:spinconnection}
D^{S}_\pm\colon S_{+}(V_{+})\to S_{+}(V_{+})\otimes V_{\pm}^{\ast}\, ,
\end{equation}

\noindent
and the associated Dirac operator:

\begin{equation}
\label{eq:Diracop}
\slashed{D}^{S}_{+} \colon S_{+}(V_{+})\to S_{-}(V_{+})\, .
\end{equation}

\noindent
The generalized scalar curvature $GS$ associated to $D^{LC}$ is obtained from a Lichnerowicz-type equation that compares the Dirac operator squared \eqref{eq:Diracop} with the rough Laplacian of $D_-^S$ in \eqref{eq:spinconnection}:

\begin{equation}
\left(\slashed{D}^{S}_{+}\right)^2\epsilon = \left( \left( D^{S}_{-}\right)^{\ast} D^{S}_{-} + GS\right)\epsilon\, ,\qquad \epsilon \in S_{+}(V_{+})\, .
\end{equation}

\noindent
Using Bismut's original Lichnerowicz-type formula for connections with skew torsion \cite{Bismut} (see also \cite{2003math......5069A,2010arXiv1012.2087F}), combined with the explicit formulae \ref{eq:Levi-Civitaexp}, it follows that:

\begin{equation}
\label{eq:GS}
4GS = S^{g} - \frac{1}{2}|H|^2,
\end{equation}

\noindent
where $S^{g}$ is the scalar curvature of the standard Levi-Civita connection on $TM$ and the scalar $|H|^2$ is the norm square of $H$ with respect to the Lorentzian metric $g$
$$
|H|^2 = \frac{1}{6}H_{ijk}H^{ijk}.
$$

We shall stress that, in generalized geometry, the curvature quantities GRic and GS are independent, that is, GS cannot be obtained from GRic by taking any natural trace.


\section{Generalized metrics and pure $\mathcal{N} = 1$ six-dimensional Supergravity}
\label{sec:GSupergravity}


In this section we describe the main result of this paper, namely, the geometrization of pure $\mathcal{N} = 1$ six-dimensional Supergravity using generalized geometry.


\subsection{Pure $\mathcal{N}=1$ six-dimensional Supergravity}
\label{sec:6dSupergravity}


Le $(M,g)$ be a six-dimensional oriented spin Lorentzian manifold. The frame bundle of $M$ admits a reduction to $SO(1,5)$, which then admits a lift to $Spin(1,5)$ since the manifold is spin. Let $\Delta_{+}$ denote the positive-chirality spinor representation of $Spin(1,5)$, which in our conventions  is a two-dimensional quaternionic representation, or equivalently, a four-dimensional complex representation with an invariant quaternionic structure\footnote{A quaternionic structure on a complex vector space $V$ is complex antilinear map $J\colon V \to V$ such that $J^{2} = -1$.}.
 
Similarly, let $s_{1}$ be the fundamental representation of the R-symmetry group of $\mathcal{N}=1$ Supergravity, which is isomorphic to $Sp(1)$. Hence, $s_{1}$ is a two-dimensional complex representation with an invariant quaternionic structure $j \colon s_{1} \to s_{1}$.

The tensor product $\Delta_{+}\otimes s_{1}$ of these two representations is an eight-dimensional complex representation with an invariant conjugation $c$ given by $c = J \otimes j$, and therefore it is a complex representation of real type. In other words, it is the complexification of a real representation $\Delta^{\mathbb{R}}_{+}$, which hence satisfies:

\begin{equation}
\Delta_{+} \otimes_{\mathbb{C}} s_{1}  \simeq\Delta^{\mathbb{R}}_{+} \otimes_{\mathbb{R}} \mathbb{C}\,.
\end{equation} 

\noindent
The real representation $\Delta^{\mathbb{R}}_{+}$ is the eight-dimensional symplectic real subspace of $\Delta_{+} \otimes_{\mathbb{C}} s_{1}$ fixed under the conjugation $c$, and it is precisely the relevant spinorial representation for the six-dimensional $\mathcal{N}=1$ Supergravity theory. We will denote by $\mathsf{S}_{+}$ the spinor bundle on $M$ associated to this representation, namely:

\begin{equation}
S_{+} \otimes_{\mathbb{C}} S_{1}  \simeq \mathsf{S}_{+} \otimes_{\mathbb{R}} \mathbb{C}\,.
\end{equation} 

\noindent
where $S_{+}$ is the positive-chirality spinor bundle over $M$ and $S_{1}$ is the vector bundle associated to a choice of $Sp(1)$-principal bundle over $M$ that corresponds to the R-symmetry group of the theory. For pure $\mathcal{N}=1$ six-dimensional Supergravity $S_1$ is trivial, endowed with the trivial connection.

After this brief \emph{detour}, based on the exposition presented in \cite{2007math......2205F,Figueroa-O'Farrill:2013aca}, we are ready to introduce the matter content of six-dimensional pure $\mathcal{N}=1$ Supergravity. This is given by 
a Lorentzian metric $g$, a self-dual closed three-form $H\in\Omega^{3}(M)$ and the gravitino or Rarita-Schwinger field $\Psi\in\Gamma(\mathsf{S}_{+}\otimes T^{\ast}M)$. Due to the fact that $H$ is self-dual
\begin{eqnarray}
H^{-} = \frac{1}{2}(H - *H) = 0\, ,
\end{eqnarray}
there is no standard action for the theory \cite{Marcus:1982yu}. The theory is thus defined through the equations of motion, which are given by \cite{Nishino:1984gk}:
\begin{equation}\label{eq:equationsofmotion}
\mathrm{Ric}^{g} = \frac{1}{4} H \circ H\, ,\qquad d\ast H = 0\, , \qquad H^{-} = 0,
\end{equation}
where
$$
(H \circ H)_{ij} = H_{ikl} H_j^{\hspace{0.1cm}kl}.
$$

\noindent
In addition to the equations of motion, the supersymmetry variation of the gravitino
\begin{equation}
\nabla_{v}\epsilon + \frac{1}{2}\iota_{v} H\cdot \epsilon\, , \qquad v\in\mathfrak{X}(M),\qquad \epsilon \in \Gamma(\mathsf{S}_{+}) ,
\end{equation} 
defines a connection on the spinor bundle $\mathsf{S}_{+}$, where $\nabla$ is the spin connection on $\mathsf{S}_{+}$ induced by the Levi-Civita connection and $\cdot$ denotes
the Clifford action of forms on spinors. This connection is induced by the metric compatible spin connection $\nabla^{+}$ with fully antisymmetric 
torsion $H$
 
\begin{equation}
\nabla^{+} = \nabla + \frac{1}{2}g^{-1}H,
\end{equation}
where $(g^{-1}H)_{ij}^{\hspace{0.2cm}l} = H_{ij}^{\hspace{0.2cm}l}$ is a section of $T^*M \otimes \mathfrak{so}(1,5)$. Using $\nabla^+$, the condition for having a supersymmetric background can be written simply as
$$
\nabla^+ \epsilon = 0.
$$
Here, $\mathfrak{so}(1,5)$ denotes the bundle of skew-symmetric endomorphisms of $TM$ with respect to the metric $g$.


\subsection{$SO(1,5) \times SO(5,1)$-metrics and self-duality}


In Euclidean gravity, an oriented Riemannian manifold satisfying the Einstein vacuum equations, typically assumed to be locally asymptotically flat, is said to be a gravitational instanton. An important class of gravitational instantons are self-dual instantons in four dimensions \cite{Hawking:1976jb}: those whose curvature tensor $R$, thought of as a bundle-valued 2-form, satisfies $R = *R$ where $*$ denotes the Hodge star operator. 
More generally, one can relax this condition and ask for the Weyl curvature to be self-dual, leading to the well-studied notion of \emph{self-dual metric}. Gravitational instantons and self-dual metrics are analogues in Euclidean gravity of instantons and self-dual gauge fields in Yang-Mills theories. In Lorentzian signature, which is the case of interest for this letter, the analogue concept 
is the so-called \emph{gravitational monopole}.

In this context, it is natural to extend the notion of self-dual gravitational instanton or gravitational monopole to the realm of generalized geometry, by finding appropriate notions of self-dual metric. 
We present now a proposal for such a notion for the specific case of an exact Courant algebroid over a six-dimensional manifold $M$ endowed with an admissible Lorentzian generalized metric. We will see in Section \ref{sec:6dsugraGG} how it naturally appears in the geometrization of pure $\mathcal{N} = 1$ Supergravity.

Let $M$ be an oriented six-dimensional manifold. Let $G$ be a generalized metric on an exact Courant algebroid $E$ over $M$ of signature $(1,5)$ or, equivalently, an $SO(1,5) \times SO(5,1)$-reduction of the bundle of $SO(6,6)$-frames of $E$. Assuming further that $G$ is admissible, the induced splitting $E \cong TM \oplus T^*M$ provides a canonical isomorphism of the $SO(6,6)$-spinor bundle\footnote{Notice that the spinor bundle $S_{6,6}$ always exists, even when $M$ is not a spin manifold.} $S_{6,6}$ of $E$ with the bundle of forms:

\begin{equation}
S_{6,6} \simeq\Lambda^*T^*M\, .
\end{equation}

\noindent
Using the orientation on $M$, following \cite{2012arXiv1203.2385B,2004math......9093G}, one can define a generalized Hodge star operator on $\Lambda^*T^*M$: Since $V_+ = \mathrm{Ker} (G - \mathbb{I})$ is isomorphic to $TM$, the orientation on $M$ induces one on $V_+$. Then we let $\{e_1, \ldots,e_6\}$ be a positive orthogonal basis of $V_+$, such that $|e_j| = \pm 1$, let 
$$
\star = - e_6 \ldots e_1 \in Cl(TM \oplus T^*M)
$$ 
--- the Clifford bundle --- and define the (generalized) Hodge star as the Clifford action of $\star$ on spinors:
$$
\star \colon \Lambda^*T^*M \to \Lambda^*T^*M.
$$
Notice that $\star^2 = 1$, so $\star$ decomposes the space of forms into its $\pm 1$-eigenspaces.
\begin{definition}
We say that a $SO(6,6)$-spinor $\alpha \in \Lambda^*T^*M$ is (anti)self-dual with respect to the admisible $SO(1,5) \times SO(5,1)$-metric $G$ if $\star \alpha = (-)\alpha$.
\end{definition}

In the metric splitting, the genereralized Hodge star agrees with its classical counterpart $*$, up to signs: if $\alpha$ has degree $j$, we have
$$
\star \alpha = (-1)^{\frac{(6-j)(5-j)}{2}+1}*\alpha.
$$
In particular, for a three-form $H \in \Omega^3(M)$, regarded as an $SO(6,6)$-spinor in the metric splitting, we have $\star H = *H$.

\begin{definition}
We say that the admissible $SO(1,5) \times SO(5,1)$-metric $G$ is self-dual if the induced curvature three-form $H \in \Omega^3(M)$, regarded as an $SO(6,6)$-spinor, is self-dual, that is, $\star H = H$.
\end{definition}

Although not completely obvious, self-duality for $G$ is natural in generalized geometry, that is, is invariant under generalized diffeomorphisms. To check this, let $B \in \Omega^2(M)$ be a closed $B$-field, inducing an automorphisms $e^B$ of $TM \oplus T^*M$ (preserving the bracket and the product)
$$
e^B \colon u + \xi \to u + \iota_u B + \xi.
$$
The transformation $e^B$ on $V_+$ leaves the three-form curvature $H$ invariant, as it preserves the bracket, but induces a non trivial action on spinors
$$
\alpha \to e^B \wedge \alpha = (1 + B \wedge B / 2 + B \wedge B \wedge B / 6) \wedge \alpha.
$$
Therefore, regarded as a $SO(6,6)$-spinor, $H$ transforms as
$$
H \to e^B \wedge H.
$$
Finally, we note that the generalized Hodge star associated to
$$
e^B V_+ = \{v + \iota_v B + g(v)| v \in TM\}
$$
is
$$
e^B \star e^{-B},
$$
which proves our claim.


\noindent
\begin{remark}
We expect that other interesting notions of generalized self-dual metric can be introduced in different dimensions for appropriate signature.
\end{remark}

By analogy with the standard nomenclature in physics, we define a \emph{generalized self-dual gravitational monopole} as follows.

\begin{definition}
Let $(E,G)$ be an exact Courant algebroid $E\to M$ over a six-dimensional manifold $M$ equipped with an admissible generalized metric $G$ of signature $(1,5)$.  We will say that $(E,G)$ is a gravitational monopole if $G$ satisfies the generalized vacuum Einstein equations:

\begin{equation}\label{eq:generalizedvacuum}
GR = 0\, , \qquad GS = 0\, .
\end{equation}

\noindent
In addition, will say that $(E,G)$ is a self-dual gravitational monopole if $(E,G)$ is a gravitational monopole with self-dual generalized metric.
\end{definition}

\noindent
We will see in Section \ref{sec:6dsugraGG} that the solutions of six-dimensional pure, $\mathcal{N}=1$ Supergravity are indeed generalized self-dual gravitational monopoles. We recall that, in generalized geometry, the curvature quantities GRic and GS are independent, that is, GS cannot be obtained from GRic by taking any sensible trace. Nevertheless, for self-dual gravitational monopoles the condition $GS = 0$ is redundant.


\begin{prop}\label{prop:redundant}
On an exact Courant algebroid, a Ricci flat generalized metric of signature $(1,5)$ which is self-dual satisfies the generalized vacuum Einstein equations.
\end{prop}

\begin{proof}
From Proposition \ref{prop:Ricci} and \eqref{eq:GS} and the definition of self-duality, by hypothesis the generalized metric $G$ satisfies
\begin{eqnarray}
\mathrm{Ric}^{g} &=& \frac{1}{4} H \circ H\nonumber,\qquad d\ast H = 0, \qquad H^- = 0
\end{eqnarray}
for the Lorentzian metric $g$ and the closed three-form $H$ defined by $G$ (see Proposition \ref{prop:ghfromVandback}). Taking trace in the first equation we obtain 
$$
S_g - \frac{1}{4}|H|^2 = 0,
$$
and using the self-duality condition
$$
|H|^2 \vol_g = H \wedge *H = H \wedge H = 0,
$$
it follows that $\mathrm{GS} = 0$.
\end{proof}

\subsection{The generalized-geometry description of six-dimensional $\mathcal{N}=1$ pure Supergravity} 
\label{sec:6dsugraGG}


In this section we will obtain the equations of motion \eqref{eq:equationsofmotion} of six-dimensional $\mathcal{N}=1$ pure Supergravity as the vanishing of the generalized Ricci tensor and scalar curvature of an admissible self-dual generalized metric. 

Consider an exact Courant algebroid $(E,\la\cdot,\cdot\ra,[\cdot,\cdot],\pi)$ over a six-dimensional oriented spin manifold $M$. We note that $E$ is determined up to isomorphism by the choice of a class $\tau \in H^3(M,\RR)$, which amounts to fix the class of the closed Supergravity three-form. 

\begin{thm}
\label{thm:equationsgeneralized}
A solution $(g,H)$ of the equations of motion \eqref{eq:equationsofmotion} of pure $\mathcal{N}=1$ six-dimensional Supergravity with $[H] = \tau$ is equivalent to a self-dual generalized $SO(1,5) \times SO(5,1)$-metric on $E$ satisfying the generalized vacuum Einstein equations: 
\begin{equation}
\label{eq:generalizedequations}
\mathrm{GRic} = 0\, , \qquad \mathrm{GS} = 0\, .
\end{equation}
\end{thm}

\begin{proof}
From Proposition \ref{prop:Ricci} and \eqref{eq:GS}, we deduce that the system \eqref{eq:generalizedequations} is equivalent to:
\begin{eqnarray}
\label{eq:equationsofmotionII}
\mathrm{Ric}^{g} &=& \frac{1}{4} H \circ H\nonumber,\qquad d\ast H = 0, \qquad S_g - \frac{1}{2}|H|^2 = 0,
\end{eqnarray}
for the Lorentzian metric $g$ and the closed three-form $H$ defined by the generalized metric $G$ on $E$ (see Proposition \ref{prop:ghfromVandback}). Note that the first and third equation are equivalent to
$$
\mathrm{Ric}^{g} = \frac{1}{4} H \circ H, \qquad |H|^2 \vol_g = H \wedge *H = 0
$$
Now, the self-duality condition on $G$ implies $H = *H$ (this condition being stronger than $|H|^2 = 0$) and hence $(g,H)$ is a solution of \eqref{eq:equationsofmotion} with $[H] = \tau$. The converse follows easily from Proposition \ref{prop:redundant}.
\end{proof}

The supersymmetry variations of six-dimensional pure, $\mathcal{N}=1$ Supergravity can be also geometrized using the canonical Levi-Civita connection $D^{LC}$. For this, given an admissible generalized metric of signature $(1,5)$, we identify the positive chirality spinor bundle $S_+$ for the induced Lorentz metric $g$ on $M$, with
$$
S_+ \cong S_+(V_+).
$$
Then, the supersymmetry variations can be written using the operator induced by
$$
D_- \colon S_+(V_+) \to S_+(V_+) \otimes V_-^*
$$
on $\mathsf{S}_+$ simply as
$$
D_- \epsilon = 0.
$$

Theorem \ref{thm:equationsgeneralized} implies that the solutions of six-dimensional pure, $\mathcal{N}=1$ Supergravity can be naturally regarded as gravitational monopoles in generalized geometry. 
A generalized gravitational monopole, where we drop the self-dual assumption, satisfies the weaker system:

\begin{eqnarray}
\label{eq:systemnull}
\mathrm{Ric}^{g} &=& \frac{1}{4} H \circ H\nonumber,\qquad d\ast H = 0, \qquad |H|^2 = 0.
\end{eqnarray}

\noindent
Therefore, the generalized vacuum Einstein equations provide a relaxed, geometric condition, for the Lorentzian metric $g$ and the three-form $H$. 
To understand the last equation

\begin{equation}
\label{eq:nullH}
|H|^2 = 0\, ,
\end{equation}

\noindent
we note that $|\cdot|^2$ is a metric, induced by $g$, of split signature $(10,10)$ on $\Lambda^3T^*M$ and therefore maximal vector spaces solving \eqref{eq:nullH} correspond to Lagrangian subbundles $L\subset\Lambda^3T^*M$. Given a Lagrangian $L$, one can find a solution of \eqref{eq:nullH} by prescribing a section of $L$. In particular, the space of (anti)self-dual three-forms $L_+$ provides a natural, canonical, choice of Lagrangian subspace which indeed correspond to the Supergravity case. According to \cite[III. 1.9.]{Chevalley} (see also \cite[Prop. 1.15]{2004math......1221G}), a general Lagrangian can always be determined in terms of $SO(10,10)$-spinors at any point as follows: it is given by a choice of $k$ self-dual $3$-forms $\theta_1, \ldots, \theta_k \in L_+$, with $0 \leq k \leq 10$, and an element
$$
\tau \in \Lambda^2 L_+.
$$
The Lagrangian determined by these data is then given by
$$
L = \{h \in \Lambda^3 T^*: e^\tau \wedge \theta_1 \wedge \ldots \wedge \theta_k \cdot h = 0\},
$$
where $\cdot$ denotes Clifford multiplication and $\wedge$ denotes the wedge product on the exterior algebra of $L_+$. For a non-self dual three-form, satisfying the null condition \eqref{eq:nullH}, the pair of three-forms $\{H,*H\}$ span an isotropic subspace (contained in a maximal Lagrangian subspace) inside $\Lambda^3T^*M$, which cuts non-trivially the Lagrangian of (anti)self-dual three-forms. It would be interesting to study further what are the theories that correspond to the relaxed geometric condition \eqref{eq:generalizedequations}, for a non-self dual metric, and if they admit some sort of supersymmetrization. At this respect, we note that choosing self-dual or anti self-dual three-forms is a matter of convention in the corresponding pure, $\mathcal{N}=1$ Supergravity theory, and condition $|H|^2 = 0$ is compatible with either choice.




The geometric condition \eqref{eq:generalizedequations} can be obtained or motivated in generalized geometry from a different point of view, namely by coupling six-dimensional $\mathcal{N}=1$ pure Supergravity to a tensor multiplet $(e^{\phi},H^-,\psi)$, where $e^{\phi}$ is the six-dimensional dilaton, $H^{-}$ is a closed anti-self-dual three-form and $\psi$ is a suitable spinor. The anti-self-dual three-form $H^{-}$ of the tensor multiplet can be combined with the self-dual three-form $H^{+}$ of the gravitational multiplet $(g,H^{+},\epsilon)$ as to obtain a closed three-form $H$, subject to no constraints. In terms of $(g,H,e^{\phi})$ the theory is formally equal to the NS-NS sector of Type-II Supergravity, the only difference being that we are here in six-dimensions. Then the geometrization of six-dimensional pure, $\mathcal{N}=1$ Supergravity coupled to a tensor multiplet can be performed by means of an exact Courant algebroid along the very same lines of Section 5.1 of reference \cite{2014CMaPh.332...89G}. The corresponding Lagrangian is then:

\begin{equation}
L = \int_M \vol_g e^\phi \mathrm{SG}(D^\phi)\, ,
\end{equation}

\noindent
which implies the equations of motion:

\begin{equation}
\mathrm{GRic}(D^\phi) = 0, \qquad \mathrm{GS}(D^\phi) = 0\, .
\end{equation}

\noindent
Here, $D^\phi$ is canonical choice of torsion-free metric connection, obtained from $D^{LC}$ by adding a \emph{Weyl term} determined by $d\phi$. For constant dilaton, $D^\phi = D^{LC}$ and the equations of motion are are simply given by

\begin{equation}
\mathrm{GRic} = 0, \qquad \mathrm{GS} = 0\, ,
\end{equation}

\noindent
which recovers the relaxed condition proposed for pure Supergravity (dropping the self-duality assumption).


\section{Supersymmetric solutions and generalized metrics}
\label{sec:susysolutions}


In this section we analyze the generalized metric corresponding to a specific supersymmetric, singular, solution of six-dimensional pure $\mathcal{N}=1$ Supergravity: the self-dual string. This is a particular example of a special class of solutions, dubbed in reference \cite{Gutowski:2003rg} as \emph{$u$-independent non-twisting}.

\subsection{General supersymmetric solutions}

The general characterization of the supersymmetric solutions of pure $\mathcal{N}=1$ six-dimensional Supergravity was obtained in reference \cite{Gutowski:2003rg} using spinorial techniques pioneered in reference \cite{Tod:1983pm} by K. P. Tod. Let us briefly review the procedure as to obtain the particular solution that we will consider in this letter. The space-time $M$ of a bosonic supersymmetric solution is assumed to be equipped with a globally defined and nowhere vanishing $\epsilon\in\Gamma\left(\mathsf{S}_{+}\right)$ generating the supersymmetry transformation and satisfying the Killing spinor equation:

\begin{equation}
\nabla^{+}\epsilon = 0\, .
\end{equation}

\noindent
The presence of such spinor implies for starters a reduction of the spin bundle from $Spin(1,5)$ to $SU(2)\ltimes\mathbb{R}^{4}$ \cite{2000math......4073B}. Let us denote by $(M,g,H)$ a supersymmetric bosonic solution admitting at least one Killing spinor. Then $M$ is equipped with a globally defined vector field $v$ and a triplet of two-forms $I^{i},\, i=1,2,3,$ such that

\begin{equation}
\mathcal{L}_{v} g = 0\, , \qquad \mathcal{L}_{v} H =0\, , \qquad v\wedge \mathcal{L}_{v} I^{i} = 0\, ,
\end{equation} 

\noindent
and

\begin{equation}
I^{i}\, I^{j} = \epsilon^{ijk}\, I^{k} - \delta^{ij}\, ,\quad \, i,j,k=1,2,3.
\end{equation}

\noindent
Therefore $v$ leaves invariant the solution and the triplet $I^{i},\, i=1,2,3,$ satisfies the algebra of the imaginary unit quaternions. 

\begin{remark}
Usually, supersymmetric solutions of a given Supergravity are classified in two different classes, time-like or null, according to the norm of $v$. In the case of six-dimensional $\mathcal{N}=1$ pure Supergravity $v$ is always null.
\end{remark}

Let $(M,g,H)$ be a $u$-independent non-twisting bosonic supersymmetric solution, as defined in \cite{Gutowski:2003rg}. By definition, this class of solutions satisfy:

\begin{equation}
\label{eq:nontwisting}
v^{\flat}\wedge dv^{\flat} = 0\, ,
\end{equation}

\noindent
or, in other words, the codimension-one distribution orthogonal to $v$ is integrable. 
Then, it is shown in \cite{Gutowski:2003rg} that one can find local coordinates $(x^{u},x^{v},x^{m}),\, m =1,\hdots,4$ in an open set $U\subset M$ such that the metric and the three-form $H$ are given by

\begin{eqnarray}
\label{eq:nontwistingsolution}
g &=& \frac{1}{h}\left( dx^{u}\otimes dx^{v} + dx^{v}\otimes dx^{u}\right) - h\, g^{4}_{mn} dx^{m}\otimes dx^{n}\nonumber\, ,\\
H &=& \ast_{4}d_{4}h - \frac{1}{h^2} dx^u\wedge dx^v\wedge d_{4}h\, .
\end{eqnarray}

\noindent
Here $g_{4} = g^{4}_{mn} dx^{m}\otimes dx^{n}$ denotes a Riemannian metric on a four-dimensional manifold $\mathcal{B}$, the so-called (local) base space of the solution, 
and $h\in C^{\infty}(U)$ is a function satisfying

\begin{equation}
\Delta_{4} h = 0\, ,\qquad \partial_{x^{u}} h = \partial_{x^{v}} h = 0\, .
\end{equation}

\noindent
In addition, we have that:

\begin{equation}
J^{i} = h^{-1} I^{i}\, , \qquad d_{4}(g_{4}J^{i}) = 0\, ,
\end{equation}

\noindent
and hence $(\mathcal{B},g_{4},J^{i}),\, i=1,2,3,$ is a four-dimensional hyper-K\"ahler manifold. We stress the fact that, given such a solution, the hyper-K\"ahler manifold $\mathcal{B}$ may be defined only locally.


\subsection{The self-dual D1-D5 string}
\label{sec:susymetricstring}


By the previous discussion, in order to find a $u$-independent non-twisted supersymmetric solution it suffices to specify

\begin{itemize}

\item A four-dimensional hyper-Kahler manifold $(\mathcal{B},g_{4},J^{i}),\, i=1,2,3$.

\item  A harmonic positive function $h$ in $(\mathcal{B},g_{4})$.

\end{itemize}

\noindent
The solution is then given by $\mathbb{R}^{2}\times\mathcal{B}$, with metric given by substituting $g_{4}$ and $h$ in equation \eqref{eq:nontwistingsolution}.
A particular solution is always given by to taking $h=1$ in any hyper-K\"ahler manifold. 

In this section we choose the simplest Hyper-K\"ahler manifold, namely we will assume that $\mathcal{B}$ is 
isometric to an open subset of $\mathbb{R}^{4}$ equipped with the euclidean flat metric. With this choice, we obtain the self-dual static string solution, whose maximally analitic extension is described in \cite{Gutowski:2003rg}. The harmonic function $h$ is then locally given by:

\begin{equation}
h = a + \frac{b}{\rho^2}\, , \qquad \rho^2 = |x|^2\, ,
\end{equation}

\noindent
where $a$ and $b$ are real constants. For simplicity we will normalize $h$ as follows:

\begin{equation}
h = 1 \pm \frac{1}{\rho^2}\, ,
\end{equation}

\noindent
where the sign $\pm$ can be chose at will. From the String Theory point of view, this solution corresponds to a D1-D5 system with constant dilaton and equal D1-D5 harmonic function, which is the harmonic function we denote by $h$. Depending on the sign, we have two possibilities, namely, the black string solution and the string solution with a naked singularity:

\begin{itemize}

\item {\bf The black string}: in this case the string-like singularity is covered by a horizon. In a local patch that covers the region from the horizon of the string to the asymptotically flat infinity, using isotropic coordinates $(t,\rho, x^{m})$ the base space is
$$
\mathcal{B}_{out} = \{x \in \mathbb{R}^4: \rho(x) > 0\}
$$ 
and we can write the solution as follows:

\begin{eqnarray}\label{eq:exterior}
g &=& \frac{1}{h}\left( dx^{u}\otimes dx^{v} + dx^{v}\otimes dx^{u}\right) - h\, \delta_{mn} dx^{m}\otimes dx^{n}\nonumber\, ,\quad h = 1 + \frac{1}{\rho^2}\, ,\\
H &=& \ast_{4}d_{4}h - \frac{1}{h^2} dx^u\wedge dx^v\wedge d_{4}h\, .
\end{eqnarray}

\noindent
The horizon lies at $\rho\to 0^{+}$ whereas the asymptotic infinity is at $\rho\to \infty$. On the other hand, the interior solution, namely the solution written in a local patch that covers from the singularity to the horizon has base space 
$$
\mathcal{B}_{in} = \{x \in \mathbb{R}^4: 0 < \rho(x) < 1\}
$$ 
and can be written as follows: 

\begin{eqnarray}\label{eq:interior}
g &=& \frac{1}{h}\left( dx^{u}\otimes dx^{v} + dx^{v}\otimes dx^{u}\right) - h\, \delta_{mn}dx^{m}\otimes dx^{n}\nonumber\, ,\quad h = 1 - \frac{1}{\rho^2}\, ,\\
H &=& \ast_{4}d_{4}h - \frac{1}{h^2} dx^u\wedge dx^v\wedge d_{4}h\, .
\end{eqnarray}

\noindent
Note that \eqref{eq:interior} is the exact same expression as \eqref{eq:exterior}, for a different choice of harmonic function $h$. The horizon lies again at $\rho\to 0^{+}$ whereas the singularity is at $\rho\to 1^{-}$. There exist also smooth coordinates in a neighborhood of the horizon as proven in \cite{Gutowski:2003rg} (therefore, the metric extends smoothly to the horizon, which is not a singular hypersurface). Notice that there is a change of signature for the metric, from signature $(1,5)$ outside the horizon to signature $(5,1)$ inside the horizon, as it is standard for this type of solutions.

\item {\bf The naked string}: the naked string does not have a horizon and can be globally written as \eqref{eq:interior}
%
with base space
$$
\mathcal{B}_{nk} = \{x \in \mathbb{R}^4: \rho(x) > 1\}.
$$ 

\noindent
The singularity is at $\rho\to 1^{+}$ and the asymptotic infinity is at $\rho\to \infty$. Since there is no horizon there is no change in the signature of the metric, which remains everywhere of signature $(1,5)$.

\end{itemize}

We analyze now the generalized metric correponding to the interior region of the black string. By constructing a exact Courant algebroid over the singularity, we will show that, as a generalized metric, the black string solution admits a smooth degenerate transition to the naked string. We construct a exact Courant algebroid $E$ over
$$
\mathbb{R}^2 \times (\mathbb{R}^4 \backslash \{0\})
$$
as follows: on $U_{in} = \mathbb{R}^2 \times \mathcal{B}_{in}$ it is given by the $H$-twisted exact Courant algebroid $TU_{in} \oplus T^*U_{in}$ with twisting three-form
$$
H_{in} = \ast_{4}d_{4}h - \frac{1}{h^2} dx^u\wedge dx^v\wedge d_{4}h, \qquad h = 1 - \frac{1}{\rho^2}
$$
Similarly, on $U_{nk} = \mathbb{R}^2 \times \mathcal{B}_{nk}$ it is given by the $H$-twisted exact Courant algebroid $TU_{nk} \oplus T^*U_{nk}$ with twisting three-form $H_{nk}$ given by the same expression. Let $V \subset \mathbb{R}^2 \times (\mathbb{R}^4 \backslash \{0\})$ be a small neighborhood of the singularity 
$$
\mathbb{R}^2 \times S^3 \subset \mathbb{R}^2 \times (\mathbb{R}^4 \backslash \{0\}).
$$
Over $V$ we consider the $H$-twisted exact Courant algebroid $TV \oplus T^*V$ with twisting three-form 
$$
H_V = \ast_{4}d_{4}h.
$$
The three local Courant algebroids glue over the intersetions $V \cap U_{in}$ and $V \cap U_{nk}$, providing a well-defined exact Courant algebroid over $\mathbb{R}^2 \times (\mathbb{R}^4 \backslash \{0\})$. Explicitely, the gluing is provided by
$$
v + \xi \in TV \oplus T^*V \to e^{-b}(v + \xi) \in TU_{in/nk} \oplus T^*U_{in/nk}
$$
where the $B$-field transformation $e^b$ is given by the two-form
$$
b = h^{-1}g_{1,1}A = h^{-1} \(dx^u \otimes dx^v - dx^v \otimes dx^u\) = h^{-1} dx^u \wedge dx^v
$$
and $A \colon T\RR^2 \to T\RR^2$, satisfying $A^2 = \mathbb{I}_{T\RR^2}$, is the standard paracomplex structure on $\RR^2$.
$$
A = \left(
\begin{array}{cc}
1 & 0\\
0 & -1
\end{array}\right).
$$
To see this, we note simply that the trasformed twisting three-form is
$$
H' = H_{in/nk} - db = \ast_{4}d_{4}h - \frac{1}{h^2} dx^u\wedge dx^v\wedge d_{4}h + \frac{1}{h^2} dx^u\wedge dx^v\wedge d_{4}h = \ast_{4}d_{4}h\, .
$$
With the exact Courant algebroid $E$ at hand, we analyze now the form of the generalized metric provided by the black string solution. On the local patch $U_{in} = \mathbb{R}^2 \times \mathcal{B}_{in}$, the corresponding generalized metric is given by
\begin{equation}
G = \left(
\begin{array}{cccc}
0 & 0 & h g_{1,1}^{-1} & 0\\
0 & 0 & 0 & h^{-1}g_4^{-1}\\
h^{-1} g_{1,1} & 0 & 0 & 0\\
0 & h g_4 & 0 & 0
\end{array}\right),
\end{equation}
where $g_4 = \delta_{mn}dx^m \otimes dx^n$. Similarly, we note that solution provided by the naked singularity gives the exact same expression on the patch $U_{nk} = \mathbb{R}^2 \times \mathcal{B}_{nk}$. On the patch $V$ around the singularity, the generalized metric takes now the form
\begin{equation}
e^{b}Ge^{-b} = \left(
\begin{array}{cccc}
A & 0 & h g_{1,1}^{-1} & 0\\
0 & 0 & 0 & h^{-1}g_4^{-1}\\
0 & 0 & A^* & 0\\
0 & h g_4 & 0 & 0
\end{array}\right).
\end{equation}
We note that the generalized metric $V_+ \subset E$ given by the previous expression, has a well-defined smooth limit for any point $p = (x,v)$ at the singularity $\mathbb{R}^2 \times S^3$, given by the subbundle
\[
\mathrm{Ker}_x(A \oplus A^* - \mathbb{I}) \oplus T_v \mathbb{R}^4 \subset \(T_x\mathbb{R}^2 \oplus T^*_x\mathbb{R}^2\) \oplus \(T_v \mathbb{R}^4 \oplus \oplus T^*_v \mathbb{R}^4\) = E_p.
\]
This subbundle corresponds to a smooth degenerate generalized metric on $E_{|\mathbb{R}^2 \times S^3}$, with signature  $(1,1)$ on the $\mathbb{R}^2$-directions. We conclude therefore that there exists a smooth generalized metric on the exact Courant algebroid $E$ over $\mathbb{R}^2 \times (\mathbb{R}^4 \backslash \{0\})$, degenerate over the singularity $\mathbb{R}^2 \times S^3$ of the black string, which interpolates between the black string solution and the naked string, as claimed.


\acknowledgments

C.S.S. would like to thank Iosif Bena, Ruben Minasian and Tom\'as Ort\'in for their interesting comments, and in particular Iosif Bena for the very nice explanations and suggestions. The work of C.S.S. was supported in part by the ERC Starting Grant 259133 -- ObservableString. The work of M.G.F. was partially supported by ICMAT Severo Ochoa project SEV-2011-0087.


\appendix


\section{Linear algebra preliminaries}
\label{app:linearalgebra}


In this section we will review the linear geometry of the fibre of an exact Courant algebroid. Let $(W,\la\cdot,\cdot\ra)$ be a $2d$-dimensional vector space equipped with a metric $\la\cdot,\cdot\ra$ of indefinite signature, and given by a short exact sequence

\begin{equation}
\label{eq:exactvector}
\xymatrix{
0 \ar[r] & V^{\ast} \ar[r]^{\pi^{\ast}} & W \ar[r]^\pi & V \ar[r] & 0.
}
\end{equation}

\noindent
Here, $V$ be a $d$-dimensional vector space and 
\begin{equation}
\pi^{\ast}\colon V^{\ast}\to W,
\end{equation}
is the map given by the dual map $\pi^{\ast}\colon V^{\ast}\to W^{\ast}$ and the isomorphism $W^{\ast} \simeq W$ provided by the metric. We assume that the image of $\pi^*$ is isotropic in $W$, that is, $\langle \pi^* \xi, \pi^*\xi\rangle = 0$ for all $\xi \in V^*$.


\noindent
An isotropic splitting of the exact sequence \eqref{eq:exactvector} is an inyective linear map $s\colon V\to W$ such that:

\begin{itemize}

\item $\pi\circ s = \mathbb{I}_{V}$

\item $\la s(v_{1}), s(v_{2})\ra = 0\, , \qquad\forall\,\, v_{1}, v_{2} \in V\, .$

\end{itemize}

\noindent
Notice that $\pi^{\ast}(V^{\ast})\cap s(V) = \left\{0\right\}$ and therefore, 
since $\dim\,s(V) = d$, the metric $\la\cdot,\cdot\ra$ must be of split signature $(d,d)$. Given an isotropic splitting $s\colon V\to W$ we can write $W$ as follows:

\begin{equation}
W\simeq s(V)\oplus \pi^{\ast}(V^{\ast}) \simeq V\oplus V^{\ast}\, ,
\end{equation}

\noindent
by means of the following isomorphism:

\begin{eqnarray}
\Psi_{s}\colon V\oplus V^{\ast} &\to & W\, ,\nonumber\\
v + \xi &\mapsto & s(v)+ \frac{1}{2}\pi^{\ast}(\xi).\,
\end{eqnarray}



\noindent
The transported metric to $V\oplus V^{\ast}$ is given by

\begin{equation}
\label{eq:transportedmetriclinear}
\la w_{1}, w_{2} \ra = \la v_{1} + \xi_{1}, v_{2} + \xi_{2}\ra = \frac{1}{2} (\iota_{v_{2}}\xi_{1} + \iota_{v_{1}}\xi_{2})\, ,\qquad w_{1}, w_{2} \in W\, ,
\end{equation}

\noindent
where slightly abusing the notation we have denoted by the same symbol the metric on $W$ an the corresponding metric on $V\oplus V^{\ast}$. Let us use a splitting to identify $(W, \la \cdot, \cdot\ra)$ with $V\oplus V^{\ast}$ equipped with the metric $\la\cdot,\cdot\ra$ given by equation \eqref{eq:transportedmetriclinear}. The isometry group of the metric $\la\cdot,\cdot\ra$ and the canonical orientation on $V\oplus V^{\ast}$ is:

\begin{equation}
SO(V\oplus V^{\ast}) \simeq SO(d,d)\, .
\end{equation}  

\noindent
Its Lie algebra is defined as usual by:

\begin{equation}
\mathfrak{so}(V\oplus V^{\ast}) = \Lambda^{2}(V\oplus V^{\ast}) = \mathrm{End}(V)\oplus \Lambda^{2} V\oplus \Lambda^{2} V^{\ast}\, ,
\end{equation}

\noindent
and thus every element $\psi \in \mathfrak{so}(V\oplus V^{\ast})$ can be written in terms of a endomorphism $A\in\mathrm{End}(V) $ of $V$, a two form $b\in \Lambda^{2} V^{\ast}$ and a bivector $\beta\in \Lambda^{2} V$. Given this decomposition, the action of an element $\psi \in\mathfrak{so}(V\oplus V^{\ast})$ on an element $v+\xi\in V\oplus V^{\ast}$ is given by

\begin{equation}
\psi \cdot (v + \xi) = \left( A\cdot v + \beta\cdot\xi\right) + \left(\iota_{v} b -A^{t}\cdot\xi\right)\, .
\end{equation}

\noindent
By exponentiation of the previous action, we obtain the orthogonal, orientation-preserving, symmetries of $V\oplus V^{\ast}$. There are two important particular cases:

\begin{itemize}

\item Exponentiation $e^{b}$ of a $b$-transformation: $e^{b}\cdot (v+  \xi) = v + (\xi+\iota_{v}b)$.

\item Exponentiation $e^{\beta}$ of a $\beta$-transformation: $e^{\beta}\cdot (v + \xi) = (v + \iota_{\xi}\beta) + \xi$.

\end{itemize}

\noindent
The space of isotropic splittings of \eqref{eq:exactvector} is an affine space modelled on $\Lambda^{2} V^{\ast}$. Therefore, any other splitting $s^{\prime}\colon V\to W$ can be written as:

\begin{equation}
s^{\prime} = s + b\, , \qquad b\in \Lambda^{2} V^{\ast}\, .
\end{equation}

\noindent
Let $\Psi_{s^{\prime}}$ be the isomorphism associated to the isotropic splitting $s^{\prime} = s + b$. Then:

\begin{equation}
\Psi_{s^{\prime}} = e^{b}\cdot\Psi_{s}\, .
\end{equation}


\renewcommand{\leftmark}{\MakeUppercase{Bibliography}}
\phantomsection
\bibliographystyle{JHEP}
\bibliography{C:/Users/cshabazi/Dropbox/Referencias/References}

\providecommand{\href}[2]{#2}\begingroup\raggedright\begin{thebibliography}{10}

\bibitem{2002math......9099H}
N.~{Hitchin}, {\it {Generalized Calabi-Yau manifolds}},  {\em Q J Math} {\bf 54
  (3)} (Sept., 2003) 281--308, [\href{http://xxx.lanl.gov/abs/math/0209}{{\tt
  math/0209}}].

\bibitem{Coimbra:2014uxa}
A.~Coimbra, C.~Strickland-Constable, and D.~Waldram, {\it {Supersymmetric
  Backgrounds and Generalised Special Holonomy}},
  \href{http://xxx.lanl.gov/abs/1411.5721}{{\tt arXiv:1411.5721}}.

\bibitem{Coimbra:2015nha}
A.~Coimbra and C.~Strickland-Constable, {\it {Generalised Structures for
  $\mathcal{N}=1$ AdS Backgrounds}},
  \href{http://xxx.lanl.gov/abs/1504.0246}{{\tt arXiv:1504.0246}}.

\bibitem{delaOssa:2014cia}
X.~de~la Ossa and E.~E. Svanes, {\it {Holomorphic Bundles and the Moduli Space
  of N=1 Supersymmetric Heterotic Compactifications}},  {\em JHEP} {\bf 1410}
  (2014) 123, [\href{http://xxx.lanl.gov/abs/1402.1725}{{\tt
  arXiv:1402.1725}}].

\bibitem{Anderson:2014xha}
L.~B. Anderson, J.~Gray, and E.~Sharpe, {\it {Algebroids, Heterotic Moduli
  Spaces and the Strominger System}},  {\em JHEP} {\bf 1407} (2014) 037,
  [\href{http://xxx.lanl.gov/abs/1402.1532}{{\tt arXiv:1402.1532}}].

\bibitem{2015arXiv150307562G}
M.~{Garcia-Fernandez}, R.~{Rubio}, and C.~{Tipler}, {\it {Infinitesimal moduli
  for the Strominger system and generalized Killing spinors}},  {\em ArXiv
  e-prints} (Mar., 2015) [\href{http://xxx.lanl.gov/abs/1503.0756}{{\tt
  arXiv:1503.0756}}].

\bibitem{2004math......1221G}
M.~Gualtieri, {\it Generalized complex geometry},  {\em Annals of Mathematics}
  {\bf 174} (2011) 75--123.

\bibitem{Coimbra:2011nw}
A.~Coimbra, C.~Strickland-Constable, and D.~Waldram, {\it {Supergravity as
  Generalised Geometry I: Type II Theories}},  {\em JHEP} {\bf 1111} (2011)
  091, [\href{http://xxx.lanl.gov/abs/1107.1733}{{\tt arXiv:1107.1733}}].

\bibitem{2014CMaPh.332...89G}
M.~{Garcia-Fernandez}, {\it {Torsion-Free Generalized Connections and Heterotic
  Supergravity}},  {\em Communications in Mathematical Physics} {\bf 332}
  (Nov., 2014) 89--115, [\href{http://xxx.lanl.gov/abs/1304.4294}{{\tt
  arXiv:1304.4294}}].

\bibitem{Coimbra:2014qaa}
A.~Coimbra, R.~Minasian, H.~Triendl, and D.~Waldram, {\it {Generalised geometry
  for string corrections}},  {\em JHEP} {\bf 1411} (2014) 160,
  [\href{http://xxx.lanl.gov/abs/1407.7542}{{\tt arXiv:1407.7542}}].

\bibitem{Pacheco:2008ps}
P.~P. Pacheco and D.~Waldram, {\it {M-theory, exceptional generalised geometry
  and superpotentials}},  {\em JHEP} {\bf 0809} (2008) 123,
  [\href{http://xxx.lanl.gov/abs/0804.1362}{{\tt arXiv:0804.1362}}].

\bibitem{Coimbra:2011ky}
A.~Coimbra, C.~Strickland-Constable, and D.~Waldram, {\it {$E_{d(d)} \times
  \mathbb{R}^+$ Generalised Geometry, Connections and M theory}},  {\em JHEP}
  {\bf 1402} (2014) 054, [\href{http://xxx.lanl.gov/abs/1112.3989}{{\tt
  arXiv:1112.3989}}].

\bibitem{Siegel:1993bj}
W.~Siegel, {\it {Manifest duality in low-energy superstrings}},
  \href{http://xxx.lanl.gov/abs/hep-th/9308133}{{\tt hep-th/9308133}}.

\bibitem{Siegel:1993th}
W.~Siegel, {\it {Superspace duality in low-energy superstrings}},  {\em
  Phys.Rev.} {\bf D48} (1993) 2826--2837,
  [\href{http://xxx.lanl.gov/abs/hep-th/9305073}{{\tt hep-th/9305073}}].

\bibitem{Jeon:2010rw}
I.~Jeon, K.~Lee, and J.-H. Park, {\it {Differential geometry with a projection:
  Application to double field theory}},  {\em JHEP} {\bf 1104} (2011) 014,
  [\href{http://xxx.lanl.gov/abs/1011.1324}{{\tt arXiv:1011.1324}}].

\bibitem{Hohm:2010pp}
O.~Hohm, C.~Hull, and B.~Zwiebach, {\it {Generalized metric formulation of
  double field theory}},  {\em JHEP} {\bf 1008} (2010) 008,
  [\href{http://xxx.lanl.gov/abs/1006.4823}{{\tt arXiv:1006.4823}}].

\bibitem{Hohm:2010xe}
O.~Hohm and S.~K. Kwak, {\it {Frame-like Geometry of Double Field Theory}},
  {\em J.Phys.} {\bf A44} (2011) 085404,
  [\href{http://xxx.lanl.gov/abs/1011.4101}{{\tt arXiv:1011.4101}}].

\bibitem{Hohm:2011zr}
O.~Hohm, S.~K. Kwak, and B.~Zwiebach, {\it {Unification of Type II Strings and
  T-duality}},  {\em Phys.Rev.Lett.} {\bf 107} (2011) 171603,
  [\href{http://xxx.lanl.gov/abs/1106.5452}{{\tt arXiv:1106.5452}}].

\bibitem{Hohm:2011dv}
O.~Hohm, S.~K. Kwak, and B.~Zwiebach, {\it {Double Field Theory of Type II
  Strings}},  {\em JHEP} {\bf 1109} (2011) 013,
  [\href{http://xxx.lanl.gov/abs/1107.0008}{{\tt arXiv:1107.0008}}].

\bibitem{Hohm:2011ex}
O.~Hohm and S.~K. Kwak, {\it {Double Field Theory Formulation of Heterotic
  Strings}},  {\em JHEP} {\bf 1106} (2011) 096,
  [\href{http://xxx.lanl.gov/abs/1103.2136}{{\tt arXiv:1103.2136}}].

\bibitem{Hohm:2013jaa}
O.~Hohm, W.~Siegel, and B.~Zwiebach, {\it {Doubled $\alpha'$-geometry}},  {\em
  JHEP} {\bf 1402} (2014) 065, [\href{http://xxx.lanl.gov/abs/1306.2970}{{\tt
  arXiv:1306.2970}}].

\bibitem{Hohm:2014eba}
O.~Hohm and B.~Zwiebach, {\it {Green-Schwarz mechanism and $\alpha'$-deformed
  Courant brackets}},  {\em JHEP} {\bf 1501} (2015) 012,
  [\href{http://xxx.lanl.gov/abs/1407.0708}{{\tt arXiv:1407.0708}}].

\bibitem{Gutowski:2003rg}
J.~B. Gutowski, D.~Martelli, and H.~S. Reall, {\it {All Supersymmetric
  solutions of minimal supergravity in six- dimensions}},  {\em
  Class.Quant.Grav.} {\bf 20} (2003) 5049--5078,
  [\href{http://xxx.lanl.gov/abs/hep-th/0306235}{{\tt hep-th/0306235}}].

\bibitem{Bena:2011dd}
I.~Bena, S.~Giusto, M.~Shigemori, and N.~P. Warner, {\it {Supersymmetric
  Solutions in Six Dimensions: A Linear Structure}},  {\em JHEP} {\bf 1203}
  (2012) 084, [\href{http://xxx.lanl.gov/abs/1110.2781}{{\tt
  arXiv:1110.2781}}].

\bibitem{Mathur:2005zp}
S.~D. Mathur, {\it {The Fuzzball proposal for black holes: An Elementary
  review}},  {\em Fortsch.Phys.} {\bf 53} (2005) 793--827,
  [\href{http://xxx.lanl.gov/abs/hep-th/0502050}{{\tt hep-th/0502050}}].

\bibitem{Lunin:2001fv}
O.~Lunin and S.~D. Mathur, {\it {Metric of the multiply wound rotating
  string}},  {\em Nucl.Phys.} {\bf B610} (2001) 49--76,
  [\href{http://xxx.lanl.gov/abs/hep-th/0105136}{{\tt hep-th/0105136}}].

\bibitem{Bena:2005va}
I.~Bena and N.~P. Warner, {\it {Bubbling supertubes and foaming black holes}},
  {\em Phys.Rev.} {\bf D74} (2006) 066001,
  [\href{http://xxx.lanl.gov/abs/hep-th/0505166}{{\tt hep-th/0505166}}].

\bibitem{Bena:2007kg}
I.~Bena and N.~P. Warner, {\it {Black holes, black rings and their
  microstates}},  {\em Lect.Notes Phys.} {\bf 755} (2008) 1--92,
  [\href{http://xxx.lanl.gov/abs/hep-th/0701216}{{\tt hep-th/0701216}}].

\bibitem{Lunin:2002iz}
O.~Lunin, J.~M. Maldacena, and L.~Maoz, {\it {Gravity solutions for the D1-D5
  system with angular momentum}},
  \href{http://xxx.lanl.gov/abs/hep-th/0212210}{{\tt hep-th/0212210}}.

\bibitem{Bena:2015bea}
I.~Bena, S.~Giusto, R.~Russo, M.~Shigemori, and N.~P. Warner, {\it {Habemus
  Superstratum! A constructive proof of the existence of superstrata}},
  \href{http://xxx.lanl.gov/abs/1503.0146}{{\tt arXiv:1503.0146}}.

\bibitem{deLange:2015gca}
P.~de~Lange, D.~R. Mayerson, and B.~Vercnocke, {\it {Structure of
  Six-Dimensional Microstate Geometries}},
  \href{http://xxx.lanl.gov/abs/1504.0798}{{\tt arXiv:1504.0798}}.

\bibitem{Hit0}
N.~Hitchin, {\it The geometry of three-forms in six dimensions},  {\em J. Diff.
  Geom.} {\bf 55} (2000) 547--576.

\bibitem{2007arXiv0710.2719G}
M.~{Gualtieri}, {\it {Branes on Poisson varieties}},  {\em The Many Facets of
  Geometry: A Tribute to Nigel Hitchin} (Oct., 2010)
  [\href{http://xxx.lanl.gov/abs/0710.2719}{{\tt arXiv:0710.2719}}].

\bibitem{DiracManifolds}
T.~J. Courant, {\it Dirac manifolds},  {\em Trans. Amer. Math.} {\bf 319}
  (1990) 631--661.

\bibitem{1995dg.ga.....8013L}
Z.-J. {Liu}, A.~{Weinstein}, and P.~{Xu}, {\it {Manin Triples for Lie
  Bialgebroids}},  {\em J. Differential Geom.} {\bf 45} (1997) 547--574.

\bibitem{LettersSevera}
P.~\u{S}evera, ``Letters to alan weinstein.''
  \url{http://sophia.dtp.fmph.uniba.sk/~severa/letters/}.

\bibitem{1999math.....10078R}
D.~{Roytenberg}, {\it {Courant algebroids, derived brackets and even symplectic
  supermanifolds}},  {\em Ph.D. thesis, University of California, Berkeley,}
  (Oct., 1999) [\href{http://xxx.lanl.gov/abs/math/9910078}{{\tt
  math/9910078}}].

\bibitem{Bressler:2002ur}
P.~Bressler and A.~Chervov, {\it {Courant algebroids}},  {\em Journal of
  Mathematical Sciences} {\bf 128, Issue 4} (2005) 3030--3053,
  [\href{http://xxx.lanl.gov/abs/hep-th/0212195}{{\tt hep-th/0212195}}].

\bibitem{Bismut}
J.~M. Bismut, {\it A local index theorem for non-k\"ahler manifolds},  {\em
  Mathematische Annalen} {\bf 284} (1989) 681--699.

\bibitem{2003math......5069A}
I.~{Agricola} and T.~{Friedrich}, {\it {On the holonomy of connections with
  skew-symmetric torsion}},  {\em ArXiv Mathematics e-prints} (May, 2003)
  [\href{http://xxx.lanl.gov/abs/math/0305069}{{\tt math/0305069}}].

\bibitem{2010arXiv1012.2087F}
A.~C. {Ferreira}, {\it {A vanishing theorem in twisted de Rham cohomology}},
  {\em Proceedings of the Edinburgh Mathematical Society (Series 2)} {\bf 56}
  (Dec., 2010) [\href{http://xxx.lanl.gov/abs/1012.2087}{{\tt
  arXiv:1012.2087}}].

\bibitem{2007math......2205F}
J.~{Figueroa-O'Farrill}, {\it {Lorentzian symmetric spaces in supergravity}},
  {\em Recent Developments in Pseudo-Riemannian Geometry} (Feb., 2007)
  [\href{http://xxx.lanl.gov/abs/math/0702205}{{\tt math/0702205}}].

\bibitem{Figueroa-O'Farrill:2013aca}
J.~Figueroa-O'Farrill and N.~Hustler, {\it {The homogeneity theorem for
  supergravity backgrounds II: the six-dimensional theories}},  {\em JHEP} {\bf
  1404} (2014) 131, [\href{http://xxx.lanl.gov/abs/1312.7509}{{\tt
  arXiv:1312.7509}}].

\bibitem{Marcus:1982yu}
N.~Marcus and J.~H. Schwarz, {\it {Field Theories That Have No Manifestly
  Lorentz Invariant Formulation}},  {\em Phys.Lett.} {\bf B115} (1982) 111.

\bibitem{Nishino:1984gk}
H.~Nishino and E.~Sezgin, {\it {Matter and Gauge Couplings of N=2 Supergravity
  in Six-Dimensions}},  {\em Phys.Lett.} {\bf B144} (1984) 187.

\bibitem{Hawking:1976jb}
S.~Hawking, {\it {Gravitational Instantons}},  {\em Phys.Lett.} {\bf A60}
  (1977) 81.

\bibitem{2012arXiv1203.2385B}
H.~{Bursztyn}, G.~R. {Cavalcanti}, and M.~{Gualtieri}, {\it {Generalized
  Kaehler geometry of instanton moduli spaces}},  {\em Communications in
  Mathematical Physics} {\bf 333} (Mar., 2015) 831--860,
  [\href{http://xxx.lanl.gov/abs/1203.2385}{{\tt arXiv:1203.2385}}].

\bibitem{2004math......9093G}
M.~{Gualtieri}, {\it {Generalized geometry and the Hodge decomposition}},  {\em
  Lecture at the String Theory and Geometry workshop, Oberwolfach.} (Sept.,
  2004) [\href{http://xxx.lanl.gov/abs/math/0409093}{{\tt math/0409093}}].

\bibitem{Chevalley}
C.~Chevalley, {\it The algebraic theory of spinors and clifford algebras},
  {\em Collected Works.} {\bf 2} (1996).

\bibitem{Tod:1983pm}
K.~Tod, {\it {All Metrics Admitting Supercovariantly Constant Spinors}},  {\em
  Phys.Lett.} {\bf B121} (1983) 241--244.

\bibitem{2000math......4073B}
R.~L. {Bryant}, {\it {Pseudo-Riemannian metrics with parallel spinor fields and
  vanishing Ricci tensor}},  {\em Seminaires \& Congres} {\bf 4} (Apr., 2000)
  53--94, [\href{http://xxx.lanl.gov/abs/math/0004073}{{\tt math/0004073}}].

\end{thebibliography}\endgroup
\label{biblio}

\end{document}